\documentclass{lmcs}
\pdfoutput=1
\usepackage[utf8]{inputenc}

\usepackage{lastpage}
\lmcsdoi{20}{3}{5}
\lmcsheading{}{\pageref{LastPage}}{}{}%
{Apr.~17,~2023}{Jul.~10,~2024}{}

\usepackage{amssymb}
\usepackage{amsmath}
\usepackage{url}
\usepackage{arrows}
\usepackage{todonotes}
\usepackage{hyperref}
\usepackage{stmaryrd}
\usepackage{setspace}

\newcommand{\lfp}{\mu}
\newcommand{\gfp}{\nu}
\newcommand{\fix}[1]{#1}
\newcommand{\sem}[1]{\llbracket#1\rrbracket}
\newcommand{\dom}{\mathit{dom}}

\newcommand{\indep}{\mathit{indep}}
\newcommand{\disjoint}{\mathit{disjoint}}
\newcommand{\conc}{\mathop{\mathbf{\!;\!}}}
\newcommand{\Eqs}{\mathcal{E}\!\mathit{qs}}
\newcommand{\Nat}{\mathbb{N}}
\newcommand{\Bool}{\mathbb{B}}
\newcommand{\B}{\mathcal{B}}
\newcommand{\E}{\mathcal{E}}
\newcommand{\F}{\mathcal{F}}
\newcommand{\Fes}{\mathcal{F}\!\mathit{es}}
\newcommand{\X}{\mathcal{X}}
\newcommand{\Val}{\mathcal{V}\!\mathit{al}}
\newcommand{\Spec}{\mathcal{S}\!\mathit{pec}}
\renewcommand{\S}{\mathcal{S}}
\renewcommand{\P}{\mathcal{P}}

\newcommand{\dep}[1]{\mathrel{\smash{\dblarrowsuper{#1}}}}
\newcommand{\depstep}[1]{\mathrel{\smash{\arrowsuper{#1}}}}
\newcommand{\depnonempty}[1]{\mathrel{\smash{\arrowsuper{#1}\hspace*{-0.4em}\raisebox{0.5em}{\small +}}}}
\newcommand{\splitname}{\mathit{split}}
\newcommand{\splitS}[3]{\splitname_{#1,#2}(#3)}
\newcommand{\restrictedeq}[1]{%
	\mathrel{\ooalign{\hfil\lower.5ex\hbox{$=$}\hfil\cr
			\hfil\raise.8ex\hbox{${\scriptstyle#1}\mkern.5mu$}\hfil\cr}}
}
\newcommand{\compl}[1]{\overline{#1}}
\newcommand{\unfold}{\mathit{unfold}}
\newcommand{\eqby}[1]{\mathrel{\smash{\stackrel{(#1)}{=}}}}

\newcommand{\invRightarrow}{\mathrel{\phantom{\Rightarrow}}}
\newcommand{\emptyeqby}[2]{\mathrel{\smash{\stackrel{\phantom{(#1)}}{#2}}}}
\newcommand{\inveq}{\mathrel{\phantom{=}}}

\begin{document}

\title{Operations on  Fixpoint Equation Systems}
\author[T. Neele]{Thomas Neele\lmcsorcid{0000-0001-6117-9129}}[a]
\author[J. van de Pol]{Jaco van de Pol\lmcsorcid{0000-0003-4305-0625}}[b]

\address{Eindhoven University of Technology, The Netherlands}
\email{t.s.neele@tue.nl}
\address{Aarhus University, Denmark}
\email{jaco@cs.au.dk}

\begin{abstract}
We study operations on fixpoint equation systems (FES) over
arbitrary complete lattices. We investigate under which conditions
these operations, such as substituting variables by their definition,
and swapping the ordering of equations, preserve the solution of a
FES. We provide rigorous, computer-checked proofs. Along the way, we
list a number of known and new identities and inequalities on extremal
fixpoints in complete lattices.
\end{abstract}

\maketitle

\section{Introduction}
This paper deals with operations on systems of fixpoint equations over
an arbitrary complete lattice. We investigate when these operations
preserve the solution of the equations. An example of a system of
equations is the set $\E:=\{X = f(X,Y,Z),\ Y=g(X,Y,Z),\ Z=h(X,Y,Z)\}$. For
most results, it is required that the functions $f,g,h$ are monotonic
in the given lattice. Such systems may well have multiple
solutions. In order to specify particular solutions, we introduce
specifications, for example $\S := [\mu X, \nu Y, \mu Z]$, indicating for
each variable whether we are interested in the minimal ($\mu$) or
maximal ($\nu$) solution.  The order of the variables in these
specifications is relevant: the leftmost variable indicates the fixpoint
with the highest priority. A {\em Fixpoint Equation System} (FES)~\cite{Mader96} is
a pair $(\E,\S)$, where $\E$ is a set of equations, and $\S$ is a
specification. Several well known instances are obtained by
instantiating the complete lattice.

\paragraph{Well-known instances of FES}
\emph{Boolean Equation Systems} (BES) arise as FES over the complete lattice
$\bot<\top$, and were proposed
in~\cite{Andersen94,AnVe95} for solving the model checking and
equivalence checking problems on finite labeled transition systems
(LTS).
BES received extensive study in~\cite{Mader96,MaSi03,GrKe04,Mateescu06}.
A generalisation to the domain $\mathbb{R} \cup \{-\infty,\infty\}$ is
\emph{real equation systems} (RES)~\cite{groote_real_2023}.

An equivalent notion to BES is two-player \emph{parity games}~\cite{Emerson1991},
see~\cite{Mader96} for a proof.
Algorithms for solving parity games receive a lot of attention, since this
is one of the few problems which is in NP and in co-NP, but not known to
be in P.
Recently, it has been shown that parity games (and thus BES) can be
solved in quasi-polynomial time~\cite{calude_deciding_2017}.
This result has also been lifted to the general setting of FES on finite
lattices~\cite{hausmann_quasipolynomial_2021,jurdzinski_universal_2022}.
Other types of games can also be seen as an instance of FES, for example
energy parity games~\cite{chatterjee_energy_2012} are FES on the lattice
$\mathbb{Z} \to \{\bot,\top\}$, ordered pointwise.
A modern parity game solver is Oink~\cite{van_dijk_oink_2018}.

\emph{Parameterised Boolean Equation Systems} (PBES, also known as
first-order, or predicate BES) arise as FES over the powerset lattice
$(2^D,\subseteq)$, with $D$ some data type, typically
representing the state space of a possibly infinite LTS.
In~\cite{Mateescu98,GrMa99}, PBES are proposed to encode the
model-checking problem of first-order mu-calculus on infinite LTSs;
they are studied in more detail in~\cite{GrWiSCP05,GrWiTCS05}. An
encoding of (branching) bisimulation of infinite LTSs in PBES is
proposed in~\cite{CPPW07}. Various procedures that operate on PBES
have been studied, for example to simplify~\cite{neele_removing_2022,Orzan2010}
or solve PBES~\cite{NeeleWG2020,neele_partial-order_2022,Ploeger2011}.
Algorithms for solving some timed fragments of PBES automatically are
studied in~\cite{ZhCl05}. PBES are implemented in the mCRL2~\cite{Bunte2019}
and CADP~\cite{Garavel2013} model checking toolsets.
MuArith~\cite{kobayashi_foldunfold_2020} is similar to PBES, but the domain $D$
is restricted to integers.

Fixpoint Equation Systems over arbitrary complete lattices (FES) are
defined in~\cite{Mader96,TaCl02}.
Some works refer to the same concept as
\emph{Hierarchical Equation Systems} (HES)~\cite{Seidl96,kobayashi_temporal_2019},
\emph{Systems of Fixpoint Equations}~\cite{baldan_abstraction_2020} or
\emph{Nested Fixpoint Equations}~\cite{jurdzinski_universal_2022}.
In~\cite{ZhCl05} it is recognized
that BES and PBES (and also \emph{Modal Equation Systems}~\cite{cleaveland_linear-time_1993}, an equational representation of the modal mu-calculus) are instances of
FES. FESs are mainly useful to provide generic definitions for all
these kinds of equation systems. We claim that the generic semantics of
a FES is more elegant than the semantics of PBES, as given in
e.g.~\cite{GrWiTCS05}. In particular, equations in FES are defined
in a semantic manner as functions on valuations, rather than on
syntactic expressions (possibly with binders).
Another advantage of FES is that one can derive
a number of basic theorems for equation systems over all lattices
in one stride, like in Chapter 3 of Mader's thesis~\cite{Mader96}. 

\emph{Abstract dependency graphs}~\cite{enevoldsen_abstract_2019} are
similar to FES, but variables range over a Noetherian partial order
with a least element, instead of a complete lattice.
When assuming every right-hand side is effectively computable, minimal
fixpoints can be computed in an iterative fashion.
Dependency graphs do not contain fixpoint alternations.

\paragraph{Contributions}
Our main goal is to study basic operations on FES, related with
substituting variables in the equations by their definition or final
solution, or swapping the order of equations in the specification.
Substitution operations form the basis of solving BES by so-called
Gauss-elimination~\cite{Mader96}. Also for PBES,
Gauss elimination plays a crucial role in their solution.
Reordering the variables in the
specification is useful, because it may give rise to independent
subspecifications that can be solved separately.  Also, swapping the
order of variables may bring down the number of alternations between
$\mu$ and $\nu$, thus lowering the complexity of certain solution
algorithms. 

Our results consist of equalities and inequalities between FES, expressing 
under which conditions the basic operations preserve the solution of a FES.
The main results are summarized in Table~\ref{mainresults} (Section~\ref{related}). 
In particular:
\begin{enumerate}
\item Results on substitution for BES and PBES are generalized to FES.
\item Results on swapping variables are generalized and sorted out, by
  weakening existing conditions, and by providing alternative conditions.
\item We provide rigorous proofs of our results.  All
  proofs in this paper have been proof-checked mechanically by the
  Coq theorem prover~\cite{bertot_short_2008,coq_software} (version 8.17)
  as well as the PVS theorem prover~\cite{mohamed_brief_2008} (version 7.1).
  Our proofs are available online~\cite{artefact}.
\end{enumerate}

\paragraph{Overview}
We first provide the basic theory of complete lattices in
Section~\ref{completelattices} and reprove all
needed facts on fixpoints, in order to present a self-contained account.
The formal definition and semantics of Fixpoint Equation
Systems is provided in Section~\ref{fes}. The proofs
(Section~\ref{subst}, \ref{swap} and~\ref{sec:independent}) are quite
elementary. They are mainly based on induction (to deal with the recursive
definition of FES semantics, Section~\ref{fesdef}) and on identities
and inequalities on fixpoints in complete lattices.
In Section~\ref{related}, we provide examples of applications of our theory
and discuss its relation with the literature.
Finally, we highlight several aspects of our Coq and PVS formalisations
in Section~\ref{coqpvs}.

\section{Fixpoint Laws in Complete Lattices}
\label{completelattices}
A {\em partial order} on a universe $U$ is a binary relation
${\leq}\subseteq U\times U$, which is {\em reflexive} ($\forall x.\,
x\leq x$), {\em anti-symmetric} ($\forall x,y.\, x\leq y \wedge y\leq
x \Rightarrow x=y$) and {\em transitive} ($\forall x,y,z.\, x\leq y
\wedge y\leq z \Rightarrow x\leq z$), where in all cases $x,y,z\in U$.

Given partial orders $(U,\leq)$ and $(V,\leq)$, we define partial
orders $(U\times V,\leq)$ and $(U\to V,\leq)$ pointwise:
$(u_1,v_1)\leq (u_2,v_2)$ iff $u_1\leq u_2\wedge v_1\leq v_2$, and
$f\leq g$ iff $\forall x\in U. f(x)\leq g(x)$. Function $f:U\to V$ is
called {\em monotonic}, iff $\forall x,y.\, x\leq y \Rightarrow f(x)\leq
f(y)$.

Given a set $X\subseteq U$, we define its set of {\em lower bounds} in $U$ as
$lb(X) := \{ y\in U ~|~ \forall x\in X.\, y\leq x\}$. If $y\in lb(X)$ and
$z\leq y$ for all $z\in lb(X)$, then $y$ is called the {\em greatest lower
  bound} of $X$. A {\em complete lattice} is a triple $(U,\leq,glb)$, where
$\leq$ is a partial order, and $glb(X)$ returns the greatest lower bound of
$X$ in $U$, for all (finite or infinite) $X\subseteq U$.

Given a complete lattice $(U,\leq,glb)$, define the partial order $(U,\geq)$,
by $x\geq y$ iff $y\leq x$. We define the set of {\em upper bounds} of
$X\subseteq U$ by $ub(X) := \{ y\in U ~|~ \forall x\in X.\, y\geq x\}$.
Define $lub(X) := glb(ub(X))$. Clearly, for all $y\in ub(X)$, $lub(X)\leq y$.
But also $lub(X)\in ub(X)$, for if $x\in X$, then $x\in lb(ub(X)$, hence
$x\leq glb(ub(X)$. So $lub(X)$ yields the {\em least upper bound} of $X$, and
$(U,\geq,lub)$ is a complete lattice as well.

Given a complete lattice $(U,\leq,glb)$, we define the {\em least
fixpoint} ($\lfp$) and {\em greatest fixpoint} ($\gfp$) 
of any function $f:U\to U$ (not only for monotonic) as follows:
\[ \begin{array}{rcl}
  \lfp(f) &:=& glb( \{ x \,|\, f(x) \leq x \} )\\
  \gfp(f) &:=& lub( \{ x \,|\, x \leq f(x) \} )\\
\end{array}\]
For $\sigma\in\{\mu,\nu\}$, we abbreviate $\sigma(\lambda x.f(x))$ by $\sigma
x.f(x)$. Note that by definition, $\gfp$ in $(U,\leq,glb)$ equals $\lfp$ in
$(U,\geq,lub)$, so theorems on $(\lfp,\leq)$ hold for $(\gfp,\geq)$ as well
``by duality''.  Also note that $F:U\to U$ is monotonic in $(U,\leq)$ if and
only if it is monotonic in $(U,\geq)$.  A direct consequence of the
definition of $\lfp$ is the following principle (and its dual):
\begin{center}
\hfill $f(x)\leq x ~~\Rightarrow~~ \lfp(f)\leq x$ \hfill ({\em $\mu$-fixpoint
  induction})\\
\hfill $x\leq f(x) ~~\Rightarrow~~ x\leq\gfp(f)$ \hfill ({\em $\nu$-fixpoint
  induction})
\end{center}

\noindent We now have the following identities on fixpoint expressions:

\begin{lem} \label{fixcalculus}
Let $(U,\leq,glb)$ be a complete lattice.
Let $\sigma\in\{\mu,\nu\}$, $A\in U$, and let $F,G\in U\to U$ and
$H,K\in U\times U\to U$ be monotonic functions. Then:
\begin{enumerate}
\item $F(\fix{\sigma}(F)) = \fix{\sigma}(F)$ \hfill (computation rule)
\item $\sigma x.\,A = A$ \hfill (constant rule)
\item $\sigma x.\,F(G(x)) = F(\sigma x.\,G(F(x)))$\hfill (rolling rule)
\item $\sigma x.\,F(F(x)) = \sigma x.\,F(x)$       \hfill (square rule)
\item $\fix{\sigma}$ is monotonic \hfill (fixpoint monotonicity)
\item $\sigma x.\,H(x,x) = \sigma x.\,\sigma y.\,H(x,y)$ \hfill (diagonal rule)
\item $\sigma x.\,H(x,x) = \sigma x.\,H(x,H(x,x))$ \hfill (unfolding rule)
\item $\sigma x.\,H(x,x) = \sigma x.\,H(x,\sigma x.\,H(x,x))$ \hfill (solve rule)
\item $\sigma x.\, H(x,\sigma y.\, K(y,x)) =
       \sigma x.\, H(x,\sigma y.\, K(y,\sigma z.\,H(z,y)))$ \hfill (Beki\v{c} rule)
\end{enumerate}
\end{lem}
\begin{proof}\mbox{ }
We first prove the theorem for $\sigma=\mu$. By the observations above,
the theorem then follows for $\sigma=\nu$ as well (``by duality'').
\begin{enumerate}
\item \begin{enumerate}
\item %1a
Let $y$ with $F(y)\leq y$ be given. Then by fixpoint induction,\\ 
$\lfp(F) \leq y$. By monotonicity, $F(\lfp(F)) \leq F(y) \leq y$. 
Since $y$ is arbitrary,
$F(\lfp(F))$ is a lower bound of $\{x\,|\,F(x) \leq x\}$.
Hence $F(\lfp(F)) \leq glb(\{x\,|\,F(x) \leq x\}) = \lfp(F)$
\item %1b
$F(\lfp(F))\leq \lfp(F)$ by (a), so by monotonicity,\\
$F(F(\lfp(F)))\leq F(\lfp(F))$. By fixpoint induction, 
$\lfp(F)\leq F(\lfp(F))$.
\end{enumerate}
Then by anti-symmetry $F(\lfp(F)) = \lfp(F)$.
\item %2
Follows directly from (1) by taking $F := \lambda x.A$ 
(which is monotonic)
\item %3
Obviously, $\lambda x.\,F(G(x))$ and $\lambda x.\,G(F(x))$ are monotonic.
\begin{enumerate}
\item % 3a
By (1), $F(G(F(\mu x.\,G(F(x)))))=F(\mu x.\,G(F(x)))$.
Hence by fixpoint induction, $\mu x.\, F(G(x)) \leq F(\mu x.\,G(F(x)))$.
\item % 3b
\begin{align*}           % I did some alignment in the math environments, mostly to increase the distance between the operators and the formulas. If you do not agree with this, let us know.
 & \invRightarrow  G(F(G(\mu x.\,F(G(x))))) \eqby{1} G(\mu x.\,F(G(x))) \\
 & \Rightarrow \mbox{~~(by fixpoint induction)} \\
 & \invRightarrow \mu x.\, G(F(x)) \leq G(\mu x.\,F(G(x)))\\
 & \Rightarrow \mbox{~~(by monotonicity)} \\
 & \invRightarrow F(\mu x.\, G(F(x))) \leq F(G(\mu x.\,F(G(x)))) \eqby{1} \mu x.\,F(G(x))
\end{align*}
\end{enumerate}
By anti-symmetry, we obtain $\mu x.\, F(G(x)) = F(\mu x.\,G(F(x)))$.
\item \begin{enumerate}
\item %4a
Using (1) twice, $F(F(\mu x.\,F(x))) = F(\mu x.\,F(x)) = \mu x.\,F(x)$.
So by fixpoint induction, $\mu x.\,F(F(x)) \leq \mu x.\,F(x)$.
\item %4b
By (3), we get $F(\mu x.\,F(F(x))) = \mu x.\,F(F(x))$.
Hence by fixpoint induction, $\mu x.\,F(x) \leq \mu x.\,F(F(x))$.
\end{enumerate}
Then by anti-symmetry, $\mu x.\,F(F(x)) = \mu x.\,F(x)$.
\item %5
Assume $f\leq g$. Let $y$ with $g(y)\leq y$ be given.
Then $f(y)\leq g(y) \leq y$, so $\lfp(f)\leq y$ by fixpoint induction. 
Since $y$ is arbitrary, $\lfp(f)$ is a lower bound for $\{x\,|\,g(x)\leq x\}$. 
By definition $\lfp(g)$ is its greatest lower bound, so $\lfp(f)\leq\lfp(g)$.
\item \begin{enumerate}
\item % 6a
\begin{align*}
 & \invRightarrow H(\mu x.\,H(x,x),\mu x.\,H(x,x)) \eqby{1} \mu x.\,H(x,x)\\
 & \Rightarrow \mbox{~~(by fixpoint induction, applied with $F:=\lambda x.\,H(x,x)$)} \\
 & \invRightarrow \mu y.\,H(\mu x.\,H(x,x),y)\leq \mu x.\,H(x,x)\\
 & \Rightarrow \mbox{~~(by fixpoint induction)} \\
 & \invRightarrow \mu x.\,\mu y.\,H(x,y) \leq \mu x.\,H(x,x)
\end{align*}
\item %6b
Let us abbreviate $A := \mu x.\,\mu y.\,H(x,y)$. Using (5) one
can show that $\lambda x.\,\mu y.\,H(x,y)$ is monotonic.
Then:
\allowdisplaybreaks
\begin{align*}
 & \invRightarrow A \eqby{1} \mu y.\,H(A,y) \eqby{1} H(A,\mu y.\,H(A,y))\\
 & \Rightarrow \mbox{~~(by congruence and both equations above)} \\
 & \invRightarrow H(A,A) = H(A,\mu y.\,H(A,y)) = A\\
 & \Rightarrow \mbox{~~(by fixpoint induction)} \\
 & \invRightarrow \mu x.\,H(x,x) \leq A = \mu x.\,\mu y.\,H(x,y)
\end{align*}
\end{enumerate}
By anti-symmetry, we indeed get: $\mu x.\,H(x,x) = \mu x.\,\mu y.\,H(x,y)$.
\item %7
Using (4) on $\lambda y.\,H(x,y)$ yields
$\mu x.\,\mu y.\,H(x,y) = \mu x.\,\mu y.\,H(x,H(x,y))$. 
%Using (5), one can show that $\lambda x.\,\mu y.\,H(x,H(x,y))$ is monotonic.
Applying (6) to both sides yields
$\mu x.\,H(x,x) = \mu x.\,H(x,H(x,x))$.
\item %8
We use (6) twice on the function $\lambda (y,x).H(x,y)$:
\allowdisplaybreaks
\begin{align*}
 & \phantom{{}\eqby{1}{}}\mu x.\,H(x,x)\\
 & \eqby{6} \mu y.\,\mu x.\,H(x,y)\\
 & \eqby{1} \mu x.\,H(x,\mu y.\,\mu x.\,H(x,y))\\
 & \eqby{6} \mu x.\,H(x,\mu x.\,H(x,x))
\end{align*}

\item %9
Define $F(y) := \mu x.\,H(x,y)$ and $G(y) := \mu x.\,K(x,y)$. Then:
\begin{align*}
& \invRightarrow \mu y.\,F(G(y)) \eqby{3} F(\mu y.\,G(F(y))) \\
& \Rightarrow \mbox{~~~~~(by definition of $F$, $G$)}\\
& \invRightarrow \mu y.\,\mu x.\,H(x,G(y)) = F(\mu y.\,\mu x.\,K(x,F(y)))\\
& \Rightarrow \mbox{~~~~~(by 6, applied to left- and right-hand side)}\\
& \invRightarrow \mu x.\, H(x,G(x)) = F(\mu y.\, K(y,F(y)))\\
& \Rightarrow \mbox{~~~~~(by definition of $F$, $G$)}\\
& \invRightarrow \mu x.\, H(x,\mu y.\, K(y,x)) = \mu x.\, H(x,\mu y.\, K(y,\mu z.\,H(z,y)))
\qedhere
\end{align*}
\end{enumerate}
\end{proof}

A careful analysis shows that all these identities can be derived in
an equational style from the identities 1, 3, 4 and 6. A natural
question is whether all true equalities (with $\mu$ as
second order operation, and variables ranging over monotonic
functions) can be derived from these four identities in an equational
manner (thus excluding the fixpoint induction rule). We don't know
the answer, but we expect that at least the equations
$\mu x.\,F(x)=\mu x.\,F^p(x)$ are needed for all primes $p$.
Results from universal algebra don't apply directly,
due to the second order nature of the fixpoint operator.

By mixing least and greatest fixpoints, we also 
obtain a number of inequalities. In particular, 4 is new,
as far as we know. Note the similarity of (4) with Beki\v{c} Rule, 
Lemma~\ref{fixcalculus}. We will call (4) Beki\v{c} Inequality.
\begin{lem}\label{fixinequalities}
Let $(U,\leq,glb)$ be a complete lattice.
Let $\sigma\in\{\mu,\nu\}$, $A\in U$, and let $F,G\in U\to U$ and
$H,K\in U\times U\to U$ be monotonic functions. Then:
\begin{enumerate}
\item $\mu(F) \leq \nu(F)$
\item\begin{enumerate}
\item $\mu x.x \leq A$
\item $A \leq \nu x.x$
\end{enumerate}
\item $\mu x.\,\nu y.\,H(x,y) \leq \nu y.\,\mu x.\,H(x,y)$
%% \item\begin{enumerate}
%% \item $\mu x.\,H(x,\nu y.\,K(y,x)) \leq
%%        \mu x.\,H(x,\nu y.\,K(y,H(x,y)))$
%% \item $\nu x.\,H(x,\mu y.\,K(y,x)) \geq
%%        \nu x.\,H(x,\mu y.\,K(y,H(x,y)))$
%% \end{enumerate}
%% \item\begin{enumerate}
%% \item $\mu x.\,H(x,\nu y.\,K(y,H(x,y)))\leq
%%        \mu x.\,H(x,\nu y.\,K(y,\mu x.\,H(x,y)))$
%% \item $\nu x.\,H(x,\mu y.\,K(y,H(x,y)))\geq
%%        \nu x.\,H(x,\mu y.\,K(y,\nu x.\,H(x,y)))$
%% \end{enumerate}
\item\begin{enumerate}
\item $\mu x.\,H(x,\nu y.\,K(y,x)) \leq
       \mu x.\,H(x,\nu y.\,K(y,\mu x.\,H(x,y)))$
\item $\nu x.\,H(x,\mu y.\,K(y,x)) \geq
       \nu x.\,H(x,\mu y.\,K(y,\nu x.\,H(x,y)))$
\end{enumerate}
\end{enumerate}
\end{lem}

\begin{proof}\hfill
\begin{enumerate}
\item $F(\nu(F)) \eqby{\ref{fixcalculus}.1} \nu(F)$, so by fixpoint
induction, $\mu(F)\leq\nu(F)$.
\item (a) $A\leq A$, hence by fixpoint induction, $\mu x.x \leq A$.
Then (b) follows by duality.
\item Define $F(x):=\nu y.\,H(x,y)$ and $G(y):=\mu x.\,H(x,y)$. 
Note that both $F$ and $G$ are monotonic, using Lemma~\ref{fixcalculus}.5. 
Then:
\allowdisplaybreaks
\begin{align*}
 & \invRightarrow \mu(F)
     \eqby{\ref{fixcalculus}.1}
   F(\mu(F))
     \eqby{\ref{fixcalculus}.1}
   H(\mu(F),F(\mu(F)))
     \eqby{\ref{fixcalculus}.1}
   H(\mu(F),\mu(F)) \\
& \Rightarrow ~~~~~\mbox{(by fixpoint induction)}\\
  & \invRightarrow G(\mu(F)) = \mu x.\,H(x,\mu(F)) \leq \mu(F) \\
& \Rightarrow ~~~~~\mbox{(monotonicity $F$)}\\
  & \invRightarrow F(G(\mu(F))) \leq F(\mu(F)) \eqby{\ref{fixcalculus}.1} \mu(F) \\
& \Rightarrow ~~~~~\mbox{(monotonicity $H$)}\\
  & \invRightarrow F(G(\mu(F))\eqby{\ref{fixcalculus}.1}
    H(G(\mu(F)),F(G(\mu(F)))) \leq
    H(G(\mu(F)),\mu(F))\eqby{\ref{fixcalculus}.1}
    G(\mu(F))\\
& \Rightarrow ~~~~~\mbox{(by fixpoint induction)}\\
  & \invRightarrow \mu(F)\leq G(\mu(F))\\
& \Rightarrow ~~~~~\mbox{(by fixpoint induction for $\nu$)}\\
  & \invRightarrow \mu(F) \leq \nu(G)
\end{align*}
   
\item (a) Define $F(y):= \mu x.\,H(x,y)$ and $G(x):=\nu y.\,K(y,x)$. 
Note that both $F$ and $G$ are monotonic, using Lemma~\ref{fixcalculus}.5. 
Then:
\begin{align*}
 & \emptyeqby{\ref{fixcalculus}.6}{\phantom{=}} \mu x.\,H(x,\nu y.\,K(y,x))\\
 & \eqby{\ref{fixcalculus}.6}
   \mu x.\,\mu z.\,H(z,\nu y.\,K(y,x))\\
 & \emptyeqby{\ref{fixcalculus}.6}{=} \mu x.\,F(G(x)) \\
 & \eqby{\ref{fixcalculus}.3}
   F(\mu x.\,G(F(x)))\\
 & \emptyeqby{\ref{fixcalculus}.6}{\leq} \mbox{~~~~~(using 1, and monotonicity of $F$)}\\
 & \emptyeqby{\ref{fixcalculus}.6}{\phantom{=}} F(\nu x.\,G(F(x)))\\
 & \emptyeqby{\ref{fixcalculus}.6}{=} F(\nu x.\,\nu y.\,K(y,F(x)))\\
 & \eqby{\ref{fixcalculus}.6}
   F(\nu y.\,K(y,F(y)))\\
 & \emptyeqby{\ref{fixcalculus}.6}{=} \mu x.\,H(x,\nu y.\,K(y,\mu x.\,H(x,y)))
\end{align*}

Then (b) follows by ``duality'' (reversing $\mu/\nu$ and ${\leq}/{\geq}$). 
More precisely, (b) is (a) in the reversed complete lattice $(U,\geq,lub)$.
\qedhere
\end{enumerate}
\end{proof}

\noindent Note that (1) and (3) are their own dual.

\section{Fixpoint Equation Systems}
\label{fes}
In this section we first formally define Fixpoint Equation Systems (FES).
We show by examples how they generalize Boolean and Predicate Equation
Systems. Subsection~\ref{semantics} introduces the semantics
of a FES by defining its solutions. 
Finally, Subsection~\ref{depgraph} defines the variable dependency graph
in a FES.

\subsection{Definition of Fixpoint Equation Systems}
\label{fesdef}
Fix a complete lattice $(U,\leq,glb)$, and a set of variables $\X$.
Throughout the paper, we assume that equality on variables is decidable.
We define
the set of {\em valuations} $\Val := \X \to U$.  For $X\in\X$, $\eta\in\Val$,
$P\in U$, we denote by $\eta[X:=P]$ the valuation that returns $P$ on $X$ and
$\eta(Y)$ on $Y\neq X$.  As any function, valuations can be ordered pointwise,
i.e. $\eta_1\leq\eta_2$ iff $\forall X\in\X.\, \eta_1(X)\leq\eta_2(X)$. Note
that valuation update is monotonic, that is, if $P\leq Q$, then
$\eta[X:=P]\leq\eta[X:=Q]$.
To indicate that two valuations agree on a set of variables $V \subseteq \X$,
we write $\eta_1 \restrictedeq{V} \eta_2$, formally defined as
$\forall X \in V. \eta_1(X) = \eta_2(X)$.
The complement of $V$ in $\X$ is denoted $\compl{V}$.

A set of {\em mutually recursive equations} is a member of $\Eqs := \Val \to
\Val$. The set $\Eqs$ is also ordered pointwise.
$\E$ is {\em monotonic} iff it is a monotonic function on $\Val$. 
Note that this semantic view on equations escapes the need to introduce
(and be limited) to a particular syntax.

\begin{exa}\label{exeqs}
Take $\X=\{X,Y,Z\}$ and $U=\Bool$, the Boolean lattice $\bot < \top$.
We write $(a,b,c)\in\Bool^3$ as a shorthand for the valuation 
$\{X{=}a,Y{=}b,Z{=}c\}$.
The system of equations $\{X=Y\wedge Z,\ Y=X\vee Z,\ Z=\neg X\}$ 
is represented in our theory as the function
\[ \B \ :=\  
   \lambda (X,Y,Z)\in \Bool^3.\, (Y\wedge Z,\ X\vee Z,\ \neg X)\enspace. \] 
It is not monotonic, because as valuations, 
$(\bot,\bot,\bot)\leq(\top,\top,\top)$,  but 
\[ \B(\bot,\bot,\bot) = (\bot,\bot,\top) \not\leq (\top,\top,\bot) =
\B(\top,\top,\top)\enspace .\]
\end{exa}

Note that $\Eqs$ is isomorphic with $\X\to\Val\to U$. This motivates
the following slight abuse of notation: Given $\E\in\Eqs$, we will often write
$\E_X(\eta)$ for $\E(\eta)(X)$. This expression denotes the definition of $X$
in $\E$, possibly depending on other variables as represented by the
valuation $\eta$.
Similar to valuations, agreement on variables from $V \subseteq \X$ is denoted
$\E_1 \restrictedeq{V} \E_2$, defined as
$\forall \eta \in \Val, X \in V. \E_1(\eta)(X) = \E_2(\eta)(X)$.

The set of {\em specifications} consists of finite lists of signed variables:
$\Spec := (\{\mu,\nu\}\times \X)^*$. Note that a specification selects a
subset of variables to be considered, assigns a fixpoint sign to these
variables, and assigns an order to these variables. We use $\sigma X$ as a
notation for $(\sigma,X)$, write $\varepsilon$ for the empty list, and use
$\conc$ for list concatenation. We will identify a singleton list with its
element. For instance, $[\mu X,\nu Y] \conc \mu Z$ denotes the specification
$[(\mu,X),(\nu,Y), (\mu,Z)]$. We define
$\dom(\S)\subseteq\X$ as the set of variables that occur in some pair in $\S$.
Decidability of $X \in \dom(\S)$ follows from finiteness of $\S$ and
decidability of equality on variables.
We define $\disjoint(\S_1,\S_2)$ iff $\dom(\S_1)\cap\dom(\S_2)=\emptyset$.

We often require that valuations or equation systems agree on the variables
in a specification.
Accordingly, we overload $\restrictedeq{\cdot}$ so that
$\eta_1 \restrictedeq{\S} \eta_2$ (\emph{resp.} $\E_1 \restrictedeq{\S} \E_2$)
is defined as $\eta_1 \restrictedeq{\S} \eta_2$ (\emph{resp.}
$\E_1 \restrictedeq{\S} \E_2$).
This also applies when a complement is involved: $\eta_1 \restrictedeq{\compl{\S}} \eta_2$ is $\eta_1 \restrictedeq{\compl{\dom(\S)}} \eta_2$.

Finally, a {\em fixpoint equation system} (FES) 
$\F$ on $(U,\X)$ is simply a pair in $\Fes := \Eqs\times\Spec$.
Before we present the semantics of FES, we first consider several instances of FES.

\begin{exa}
The Boolean Equation System~\cite{Mader96} traditionally written as
\[\begin{array}{rcl}
\mu X & = & Y\wedge Z\\
\nu Y & = & X\vee Z\\
\nu Z & = & \neg X
\end{array}\]
is represented in our theory as the pair $(\B,[\mu X,\nu Y,\nu Z])$,
where $\B$ is from Example~\ref{exeqs}.
Note that this notation for BES integrates the set of equations and the specification into one, and these cannot be considered separately.
\end{exa}

\begin{exa}
A PBES (parameterized BES \cite{GrWiTCS05}, or predicate equation system
\cite{ZhCl05}) is a FES over the complete lattice 
$U:=(\P(D),\subseteq)$ for some data set $D$, or equivalently $(D \to \Bool, \leq)$.
PBES thus generalise BES: each variable is now a predicate over domain $D$, allowing one to create complex expressions over data.
For our example, let $\X = \{X,Y\}$ and $D = \Nat \times \Bool$.
Again using the shorthand $(X,Y) \in (D \to \Bool)^2$ for valuations, $(\B',[\mu X, \nu Y])$ is a PBES, where
\[\begin{array}{rcl}
	\B' & := & \lambda (X,Y) \in (D \to \Bool)^2.\\
	&& \quad(\lambda m\in \Nat,b \in \Bool.\
	(     b\rightarrow m>0\wedge Y(m,\bot)) \wedge 
	(\neg b\rightarrow m<5\wedge Y(m,\bot)),\\
	&& \quad\phantom{(}\lambda m\in \Nat,b \in \Bool.\ X(m-1,m>4) \vee Y(m+1,\bot))
\end{array}\]
The function on the last line, which defines $Y$, does not contain an occurrence of the argument $b$.
However, in our theory we are required to include it so that both variables in $\X := \{X,Y\}$ are predicates over the same $D:=\Nat\times \Bool$.
In the notation of~\cite{GrWiTCS05}, the same PBES is simply written as a pair of predicate definitions with accompanying fixpoint signs.
The argument $b$ of $Y$ may be left out:
\[\begin{array}{rcl}
\mu X(m\colon\Nat,b\colon\Bool) & = &
  (     b\rightarrow m>0\wedge Y(m)) \wedge 
  (\neg b\rightarrow m<5\wedge Y(m))\\
\nu Y(m\colon\Nat) &= & X(m-1,m>4) \vee Y(m+1)
\end{array}\]
Similar to BES, the PBES formalism as defined in~\cite{GrWiTCS05} does not consider the equations and the specification separately.
%% Given $\phi,\psi\in\P(D)$, we denote the valuation $\{X=\phi,Y=\psi\}$
%% by $(\phi,\psi)$. Then the PBES above is represented in our theory by
%% the FES $(\E,\S)$, where $\S := [\mu X,\nu Y]$, and $\E$ is defined by
%% \[\begin{array}{rl}
%% \lambda (X,Y)\in\P(D)^2 .
%%   \Big( &\!\! \lambda (m,b).(b\rightarrow m>0\wedge Y(m,\bot)) \wedge 
%%                         (\neg b\rightarrow m<5\wedge Y(m,\bot)),\\
%%         &\!\! \lambda (m,b).X(m-1,m>4)\vee Y(m+1,\bot)\Big)\enspace.
%% \end{array}\]
%% This is of type $(D\to\B)^2\to(D\to B)^2$, which is indeed
%% isomorphic to $(\X\to \P(D))\to(\X\to \P(D))$.
\end{exa}

\begin{rem}
	Our choice of separating the set of equations and the specification
	makes it easier to perform induction proofs over the specification
	(because one retains knowledge of all equations in the proof scope).
	However, we have not required that all variables in $\S$ are unique.  This is
	not needed in our formalization, because in Lemma~\ref{sanity}.3 we will
	show that if $X\in\dom(\S)$, then any set of equations $\E$ has the same
	semantics when combined with $S$ or with $\sigma X\conc S$, for any $\sigma \in \{\mu,\nu\}$.
	Do note that there is a hidden
	assumption: even if $X$ occurs multiple times in $\S$, possibly with
	different signs, there can only be one defining equation for it,
	because $\E$ is a function. So when transferring the results to
	traditional notation, one should add the (quite natural) requirement
	that all equations have unique variable names.
\end{rem}

\begin{exa}
\emph{Modal equation systems} (MES)~\cite{cleaveland_linear-time_1993} is a very similar FES instance to PBES since it also uses the powerset lattice $(\P(D),\subseteq)$.
However, a MES is interpreted on a \emph{labelled transition system} (LTS), so that $D$ is equal to the set of states in the LTS.
MES includes, for every action $a$ in the LTS, the modal operators $[a]\varphi$ (``$\varphi$ must hold after every possible transition labelled with $a$'') and $\langle a \rangle \varphi$ (``there exists an $a$-transition after which $\varphi$ holds'').
MES is an equational representation of the modal mu-calculus.
Adopting the notation from the previous examples, an example of a MES is:
\[\begin{array}{rcl}
	\nu X & = & Y\\
	\mu Y & = & [a] X \land [b] Y
\end{array}\]
This MES expresses that on every path in the LTS, action $a$ occurs infinitely often (and $b$ may happen only finitely often in between).
\end{exa}

\subsection{Semantics of FES and Basic Results}\label{semantics}
Next, the semantics of a FES, $\sem{\E,\S}:\Val\to\Val$, is 
defined recursively on $\S$:
\[\left\{\begin{array}{rcl}
\sem{\E,\varepsilon}(\eta)   & := & \eta \\
\sem{\E,\sigma X\conc\S}(\eta) & := & 
   \sem{\E,\S}(\eta[X:=\sigma(F)]),\\
&& \mbox{~~where } F \colon U \to U \mbox{~is defined as~}\\
&& \; F(P) := \E_X(\sem{\E,\S}(\eta[X:= P]))
\end{array}\right.\]

We now state the first results on the semantics of FES.
The lemma below states monotonicity properties for
the semantics. They ensure that fixpoints are well-behaved.
Below, recall that we use pointwise ordering, \emph{e.g.},
$\E_1 \leq \E_2$ iff $\E_1(\eta)(X) \leq \E_2(\eta)(X)$ for all $\eta$ and $X$.

\begin{lem}\label{monotonic}
Let $\E$, $\E_1$, $\E_2 \in \Eqs$ and $\S\in\Spec$, then
\begin{enumerate}
\item If $\E$ is monotonic, then $\sem{\E,\S}$ is monotonic.
\item If $\E_1$ is monotonic and $\E_1\leq \E_2$, 
      then $\sem{\E_1,\S}\leq \sem{\E_2,\S}$.
\end{enumerate}
\end{lem}
\begin{proof}\hfill
\begin{enumerate}
\item We prove $\forall \eta_1,\eta_2.\
\eta_1\leq\eta_2\Rightarrow \sem{\E,\S}(\eta_1)\leq\sem{\E,\S}(\eta_2)$ 
by induction on $\S$. The base case: assume $\eta_1\leq\eta_2$,
then
$\sem{\E,\varepsilon}(\eta_1)=\eta_1\leq\eta_2=\sem{\E,\varepsilon}(\eta_2)$.
Induction step: Let $\eta_1\leq \eta_2$ be given, then for any $P\in U$, 
we have $\eta_1[X:=P]\leq{\eta_2[X:=P]}$. So
\allowdisplaybreaks
\begin{align*}
 & \invRightarrow \mbox{~~~~~(induction hypothesis)}\\
 & \invRightarrow \sem{\E,\S}(\eta_1[X:=P])\leq\sem{\E,\S}(\eta_2[X:=P])\\
 & \Rightarrow \mbox{~~~~~($\E$ is monotonic)}\\
 & \invRightarrow \E_X(\sem{\E,\S}(\eta_1[X:=P])) \leq \E_X(\sem{\E,\S}(\eta_2[X:=P]))\\
 & \Rightarrow \mbox{~~~~~(fixpoint monotonicity, Lemma~\ref{fixcalculus}.5)}\\
 & \invRightarrow \sigma P.\,\E_X(\sem{\E,\S}(\eta_1[X:=P])) \leq 
   \sigma P.\,\E_X(\sem{\E,\S}(\eta_2[X:=P]))\\
 & \Rightarrow \mbox{~~~~~(define 
        $F_i(P) := \E_X(\sem{\E,\S}(\eta_i[X:=P]))$, for $i=1,2$)}\\
 & \invRightarrow \eta_1[X:=\sigma(F_1)] \leq \eta_2[X:=\sigma(F_2)]\\
 & \Rightarrow \mbox{~~~~~(induction hypothesis)}\\
 & \invRightarrow \sem{\E,\S}(\eta_1[X:=\sigma(F_1)]) \leq 
   \sem{\E,\S}(\eta_2[X:=\sigma(F_2)])\\
 & \Leftrightarrow \mbox{~~~~~(definition)} \\
 & \invRightarrow \sem{\E,\sigma X\conc\S}(\eta_1) \leq 
   \sem{\E,\sigma X\conc\S}(\eta_2)
\end{align*}
\item % monotonicity 2
Assume $\E_1$ is monotonic, and $\E_1\leq \E_2$ (pointwise). We prove 
by induction on $\S$ that
$\forall \eta.\,\sem{\E_1,\S}(\eta)\leq\sem{\E_2,\S}(\eta)$.
The base case is simple: for all $\eta$, we have
$\sem{\E_1,\varepsilon}(\eta)=\eta=\sem{\E_2,\varepsilon}(\eta)$.
For the induction step, let $\eta$ be given, then for any $P\in U$:
\allowdisplaybreaks
\begin{align*}
 & \invRightarrow \mbox{~~~~~(induction hypothesis)}\\
 & \invRightarrow \sem{\E_1,\S}(\eta[X:=P])\leq\sem{\E_2,\S}(\eta[X:=P])\\
 & \Rightarrow \mbox{~~~~~(monotonicity of $\E_1$ and $\E_1\leq \E_2$, pointwise))}\\
 & \invRightarrow \E_{1,X}(\sem{\E_1,\S}(\eta[X:=P]))\leq\E_{2,X}(\sem{\E_2,\S}(\eta[X:=P]))\\
 & \Rightarrow \mbox{~~~~~(fixpoint monotonicity, Lemma~\ref{fixcalculus}.5)}\\
 & \invRightarrow \sigma P.\,\E_{1,X}(\sem{\E_1,\S}(\eta[X:=P]))\leq
   \sigma P.\,\E_{2,X}(\sem{\E_2,\S}(\eta[X:=P]))\\
 & \Rightarrow
   \mbox{~~~~~(define $F_i(P) := \E_{i,X}(\sem{\E_i,\S}(\eta[X:=P]))$,
               for $i=1,2$)}\\
 & \invRightarrow \eta[X:=\sigma(F_1)] \leq \eta[X:=\sigma(F_2)] \\
 & \Rightarrow
   \begin{array}[t]{rl}
         & \sem{\E_1,\S}(\eta[X:=\sigma(F_1)])\\
    \leq & \mbox{~~~~~(monotonicity, part 1 of this lemma)}\\
         &\sem{\E_1,\S}(\eta[X:=\sigma(F_2)]) \\
    \leq & \mbox{~~~~~(induction hypothesis)} \\
         &\sem{\E_2,\S}(\eta[X:=\sigma(F_2)]) \\
   \end{array}\\
 & \Rightarrow \mbox{~~~(definition)} \\
 & \invRightarrow  \sem{\E_1,\sigma X\conc\S}(\eta) \leq 
   \sem{\E_2,\sigma X\conc\S}(\eta)
\qedhere
\end{align*}
\end{enumerate}
\end{proof}

The next lemma states three sanity properties of the semantics:
(1) The semantics only modifies the valuation on elements
in the domain; (2) the semantics only depends on equations mentioned
in the domain; (3) the input valuation is only used for variables
outside the domain.

\begin{lem}\label{sanity} Let $\E,\E_1,\E_2\in\Eqs$, 
$\S\in\Spec$, $\eta\in\Val$ and $X\in\X$.
\begin{enumerate}
\item If $X\notin\dom(\S)$ then $\sem{\E,\S}(\eta)(X) = \eta(X)$.
\smallskip
\item If $\E_1 \restrictedeq{\S} \E_2$, then 
$\sem{\E_1,\S}=\sem{\E_2,\S}$.
\smallskip
\item If $\eta_1 \restrictedeq{\compl{\S}} \eta_2$,
then $\sem{\E,\S}(\eta_1)=\sem{\E,\S}(\eta_2)$.
\end{enumerate}
\end{lem}
\begin{proof}\hfill
\begin{enumerate}
\item Induction on $\S$. The base case holds by definition. For the induction
step, assume $X\notin\dom(\sigma Y\conc S)$. Then
\begin{align*}
& \inveq \sem{\E,\sigma Y\conc\S}(\eta)(X)\\
& = \mbox{~~~~~(define $F(P) := \E_Y (\sem{\E,\S}(\eta[Y:=P])$)}\\
& \inveq \sem{\E,\S}(\eta[Y:= \sigma(F)])(X)\\
& = \mbox{~~~~~(induction hypothesis, note: $X\notin\dom(\S)$)}\\
& \inveq \eta[Y:=\sigma(F)](X)\\
& = \mbox{~~~~~($X\notin\dom(\sigma Y\conc S)$, so $X\neq Y$)}\\
& \inveq \eta(X)
\end{align*}
\item Induction on $\S$. The base case holds by definition. For the
induction step, let $\eta$ be given, and assume 
$\E_1 \restrictedeq{\sigma X\conc \S} \E_2$. 
Note that this implies $\E_1 \restrictedeq{\S} \E_2$, so
we can use the induction hypothesis. For any $P\in U$ we have:
\allowdisplaybreaks
\begin{align*}
& \invRightarrow \mbox{~~~~~(induction hypothesis)}\\
& \invRightarrow \sem{\E_1,\S}(\eta[X:=P]) = \sem{\E_2,\S}(\eta[X:=P])\\
& \Rightarrow \mbox{~~~~~($X\in\dom(\sigma X\conc\S)$, so $\E_{1,X}=\E_{2,X}$
 by assumption)}\\
& \invRightarrow \E_{1,X}(\sem{\E_1,\S}(\eta[X:=P])) = \E_{2,X}(\sem{\E_2,\S}(\eta[X:=P]))\\
& \Rightarrow 
   \mbox{~~~~~(define $F_i(P) := \E_{i,X}(\sem{\E_i,\S}(\eta[X:=P]))$
     for $i=1,2$)}\\
& \invRightarrow \eta[X:=\sigma(F_1)] = \eta[X:=\sigma(F_2)]\\
& \Rightarrow \mbox{~~~~~(induction hypothesis)}\\
& \invRightarrow \sem{\E_1,\S}(\eta[X:=\sigma(F_1)])=\sem{\E_2,\S}(\eta[X:=\sigma(F_2)])\\
& \Leftrightarrow \mbox{~~~~~(definition)}\\
& \invRightarrow \sem{\E_1,\sigma X\conc\S}(\eta)=\sem{\E_2,\sigma X\conc\S}(\eta)\\
\end{align*}
\item Induction on $\S$. The base case is trivial, because both
the assumption and the conclusion reduce to $\eta_1=\eta_2$.
The induction step is proved as follows (for arbitrary $P\in U$):
\begin{align*}
 & \invRightarrow \mbox{~~~~~(assumption)}\\
 & \invRightarrow \eta_1\restrictedeq{\compl{\sigma X\conc \S}}\eta_2\\
 & \Rightarrow \mbox{~~~~~(extensionality)} \\
 & \invRightarrow \eta_1[X:=P]\restrictedeq{\compl{\S}}
   \eta_2[X:=P]\\
 & \Rightarrow \mbox{~~~~~(induction hypothesis)}\\
 & \invRightarrow \sem{\E,\S}(\eta_1[X:=P]) 
 = \sem{\E,\S}(\eta_2[X:=P]) \\
 & \Rightarrow \mbox{~~~~~(define $F_i(P) := 
        \E_X(\sem{\E,\S}(\eta_i[X:=P]))$, for $i=1,2$)}\\
 & \invRightarrow \sigma(F_1) = \sigma(F_2)\\
 & \Rightarrow \mbox{~~~~~(extensionality)}\\
 & \invRightarrow \eta_1[X:=\sigma(F_1)]\restrictedeq{\compl{\S}}
   \eta_2[X:=\sigma(F_2)]\\
 & \Rightarrow \mbox{~~~~~(induction hypothesis)}\\
 & \invRightarrow \sem{\E,\S}(\eta_1[X:=\sigma(F_1)]) 
 = \sem{\E,\S}(\eta_2[X:=\sigma(F_2)]) \\
 & \Rightarrow \mbox{~~~~~(definition)}\\
 & \invRightarrow \sem{\E,\sigma X\conc\S}(\eta_1) 
 = \sem{\E,\sigma X\conc\S}(\eta_2)
\qedhere
\end{align*}
\end{enumerate}
\end{proof}

Next, we can prove that the semantics as defined above indeed 
solves the equations for those variables occurring in the specification:
\begin{lem}\label{solution}
Let $\E\in\Eqs$ be monotonic, $\S\in\Spec$, $\eta\in\Val$ and $X\in\X$.\\
If $X\in\dom(\S)$, 
then $\E_X(\sem{\E,\S}(\eta)) = \sem{\E,\S}(\eta)(X)$.
\end{lem}
\begin{proof}
Let $\E$ be monotonic.
By induction on $\S$, we will prove that for all
$X\in\dom(\S)$ and for all $\eta$, it holds that $\E_X(\sem{\E,\S}(\eta)) = \sem{\E,\S}(\eta)(X)$.
The base case trivially holds, because $X\notin\dom(\varepsilon)$.
For the induction step, assume $X\in\dom(\sigma Y\conc\S)$. 
We distinguish cases. 

If $X\in\dom(\S)$:
\allowdisplaybreaks
\begin{align*}
  & \inveq \sem{\E,\sigma Y\conc\S}(\eta)(X)\\
  & = \mbox{~~~~~(define $F(P) := \E_Y(\sem{\E,\S}(\eta[Y:=P]))$~)}\\
  & \inveq \sem{\E,\S}(\eta[Y:=\sigma(F)])(X)\\
  & = \mbox{~~~~~(induction hypothesis; $X\in\dom(\S)$)}\\
  & \inveq\E_X(\sem{\E,\S}(\eta[Y:=\sigma(F)]))\\
  & = \mbox{~~~~~(definition)}\\
  & \inveq \E_X(\sem{\E,\sigma Y\conc\S}(\eta))
\end{align*}

Otherwise, if $X\notin\dom(\S)$ then $X=Y$. We compute:
\begin{align*}
  & \inveq \sem{\E,\sigma X\conc\S}(\eta)(X)\\
  & = \mbox{~~~~~(define $F(P) := \E_X(\sem{\E,\S}(\eta[X:=P]))$~)}\\
  & \inveq \sem{\E,\S}(\eta[X:=\sigma(F)])(X)\\
  & = \mbox{~~~~~(Lemma~\ref{sanity}.1; $X\notin\dom(\S)$)}\\
  & \inveq \sigma(F)\\
  & =\begin{array}[t]{l}\mbox{~~~~~(computation rule, Lemma~\ref{fixcalculus}.1; 
               $F$ is monotonic because $\E_X$ }\\
   \quad\mbox{is monotonic by assumption, 
     and $\sem{\E,\S}$ is by Lemma~\ref{monotonic}.1) }\end{array}\\
  & \inveq F(\sigma(F)) \\
  & = \E_X(\sem{\E,\S}(\eta[X:=\sigma(F)]))\\
  & = \mbox{~~~~~(definition)}\\
  & \inveq \E_X(\sem{\E,\sigma X\conc\S}(\eta))
\qedhere
\end{align*}
\end{proof}

\noindent We have the following left-congruence result:
\begin{lem}\label{congruence}
For $\E_1,\E_2\in\Eqs$, $\S,\S_1,\S_2\in\Spec$, 
if $\E_1\restrictedeq{\S}\E_2$ and
$\sem{\E_1,\S_1}=\sem{\E_2,\S_2}$, then
$\sem{\E_1,\S\conc\S_1}=\sem{\E_2,\S\conc\S_2}$.
\end{lem}
\begin{proof}
Induction on $\S$. The base case is trivial.
For the induction step, assume $\sem{\E_1,\S_1}=\sem{\E_2,\S_2}$
and $\E_1\restrictedeq{\sigma X\conc \S}\E_2$.
Then also $\E_1\restrictedeq{\S}\E_2$. So, for any $P\in U$:
\begin{align*}
& \invRightarrow \mbox{~~~~~(induction hypothesis)}\\
& \invRightarrow \sem{\E_1,\S\conc\S_1}(\eta[X:=P])=
  \sem{\E_2,\S\conc\S_2}(\eta[X:=P])\\
& \Rightarrow \mbox{~~~~~(define $F_i(P) :=
  \E_{i,X}(\sem{\E_i,\S\conc\S_i}(\eta[X:=P]))$; note $\E_{1,X}=\E_{2,X}$)}\\
& \invRightarrow \eta[X:=\sigma(F_1)] = \eta[X:=\sigma(F_2)]\\
& \Rightarrow \mbox{~~~~~(induction hypothesis)}\\
& \invRightarrow \sem{\E_1,\S\conc\S_1}(\eta[X:=\sigma(F_1)]) = 
  \sem{\E_2,\S\conc\S_2}(\eta[X:=\sigma(F_2)])\\
& \Rightarrow \mbox{~~~~~(definition)}\\
& \invRightarrow \sem{\E_1,\sigma X\conc\S\conc\S_1}(\eta) = 
  \sem{\E_2,\sigma X\conc\S\conc\S_2}(\eta)
\qedhere
\end{align*}
\end{proof}

Remarkably, right-congruence doesn't hold in general.
Corollary~\ref{congright} will state a sufficient condition for
right-congruence.

\subsection{The Dependency Graph between Variables}
\label{depgraph}

Since we introduced a semantic notion of equations, avoiding syntactic
expressions, we also need a semantic notion of dependence between variables.
Given $V_1,V_2\subseteq\X$, we define that $V_1$ is {\em independent} of $V_2$
with respect to $\E$, notation $\indep(\E,V_1,V_2)$, as follows: $$\forall \eta_1,\eta_2.\,
(\eta_1\restrictedeq{\compl{V_2}} \eta_2) \Rightarrow
(\E(\eta_1)\restrictedeq{V_1} \E(\eta_2))\enspace.$$  That is: the solution
of variables $X\in V_1$ is the same for all those $\eta$ that differ
at most on the values assigned to $Y \in V_2$.  This is slightly more liberal
than the usual syntactic requirement that $Y$ {\em doesn't occur syntactically} 
in $\E_X$.
We overload the definition of $\indep$ for individual variables and
specifications (and any combination of those), \emph{e.g.}, 
$\indep(\E,\S_1,\S_2) = \indep(\E,\dom(\S_1),\dom(\S_2))$ and 
$\indep(\E, X, Y) = \indep(\E, \{X\},\{Y\})$.

This notion gives rise to the {\em variable dependency graph} of a
FES $(\E,\S)$. The variables in $\dom(\S)$ form the nodes of this graph; the edges
$X\depstep{\E,\S}Y$ are defined as $\neg\indep(\E,X,Y)$.
We define $X$ {\em depends (indirectly)} on $Y$ (written $X\dep{\E,\S}Y$) as the
reflexive, transitive closure of $\depstep{\E,\S}$. In other words,
there exists a path in the dependency graph from $X$ to $Y$. We also
use the notation $X\depnonempty{\E,\S}Y$ to denote the transitive
closure of $\depstep{\E,\S}$, i.e.~there is a {\em non-empty} path
from $X$ to $Y$ in the dependency graph.
We assume that $\indep(\E,X,Y)$ is decidable for all $\E$, $X$ and $Y$.
Since $\dom(\S)$ is finite, this also makes $\dep{\E, \S}$ and $\depnonempty{\E,\S}$ decidable.

\section{Substituting in FES Equations}\label{subst}
In this section, we define two substitution operations on the
equations of a FES, and study under which conditions these operations
preserve solutions.  The first operation allows substituting variables
by their definition.  We show that this substitution preserves
solutions in some new cases (cf. Section~\ref{related}). The second
operation replaces a variable in an equation by its solution.

\subsection{Unfolding Definitions}

We define $\unfold(\E,X,Y)$, where each occurrence of $Y$ in the
definition of $X$ is replaced by the definition of $Y$, as follows:
\[ \unfold(\E,X,Y)(\eta) := \E(\eta) [X := \E_X(\eta[Y:=\E_Y(\eta)])] \]

We will use the following observation several times. It basically states
that unfolding $Y$ in $X$ doesn't affect other equations than that for $X$.

\begin{lem}\label{unfoldaux}
If $X\notin\dom(\S)$, then 
\begin{enumerate}
\item $\unfold(\E,X,Y) \restrictedeq{\S} \E$
\item $\sem{\unfold(\E,X,Y),\S} = \sem{\E,\S}$
\end{enumerate}
\end{lem}
\begin{proof}
(1) holds by definition of $\unfold$, as it only modifies the value
of $\E$ on variable $X$. Then (2) holds by Lemma~\ref{sanity}.2.
\end{proof}

It is known (cf.~Example~\ref{exunfold}) that in general one should
not unfold $Y$ in $X$, if $Y$ precedes $X$ in the specification. As a
new result we show that we can substitute $X$ in its own definition:

\begin{lem}\label{unfoldsame}
Let $\E\in\Eqs$ be monotonic. Let $\S=\sigma X\conc\S_1$,
and $X\notin\dom(\S_1)$.
Then $\sem{\unfold(\E,X,X),\S}=\sem{\E,\S}$.
\end{lem}
\begin{proof}
For arbitrary $\eta$, define
\begin{align*}
 F(P)   &:= \E_X\Big(\sem{\E,\S_1}(\eta[X:=P])\Big)\\
 G(P)   &:= \unfold(\E,X,X)_X\Big(\sem{\E,\S_1}(\eta[X:=P])\Big) \\
 H(P,Q) &:= \E_X\Big(\sem{\E,\S_1}(\eta[X:=P]) [X:=Q]\Big)
\end{align*}
By Lemma~\ref{sanity}.1 and $X\notin\dom(\S_1)$,
we obtain: $\sem{\E,\S_1}(\eta[X:=P])(X) = \eta[X:=P](X) = P$, so 
\begin{equation}
	\sem{\E,\S_1}(\eta[X:=P])[X:=P] ~=~ \sem{\E,\S_1}(\eta[X:=P])\tag{*}
\end{equation}

Next, we prove that $\sigma(F)=\sigma(G)$:
\allowdisplaybreaks
\begin{align*}
    & \inveq \sigma P.\,G(P)\\
    & = \mbox{~~~~~(by definition of $\unfold$)}\\
    & \inveq \sigma P.\, \E_X\Big(\sem{\E,\S_1}(\eta[X:=P]) 
           \Big[X := \E_X\big(\sem{\E,\S_1}(\eta[X:=P])\big)\Big]\Big)\\
    & = \mbox{~~~~~(by * above)}\\
    & \inveq \sigma P.\, \E_X\Big(\sem{\E,\S_1}(\eta[X:=P]) 
             \Big[X := \E_X\big(\sem{\E,\S_1}(\eta[X:=P])[X:=P]\big)\Big]\Big)\\
    & = \sigma P.\,H(P,H(P,P))\\
    & = \mbox{~~~~~(unfold rule, Lemma~\ref{fixcalculus}.7)}\\
    & \inveq \sigma P.\,H(P,P)\\
    & = \mbox{~~~~~(by * above)}\\
    & \inveq \sigma P.\,F(P)
\end{align*}

We can now finish the proof:
\begin{align*}
  & \inveq \sem{\unfold(\E,X,X),\sigma X\conc\S_1}(\eta)\\
  & = \mbox{~~~~~(by definition and Lemma~\ref{unfoldaux}.2)}\\
  & \inveq \sem{\E,\S_1}(\eta[X:=\sigma(G)])\\
  & = \mbox{~~~~~(by the computation before)}\\
  & \inveq \sem{\E,\S_1}(\eta[X:=\sigma(F)])\\
  & = \mbox{~~~~~(by definition)}\\
  & \inveq \sem{\E,\sigma X\conc\S_1}(\eta)
\qedhere
\end{align*}
\end{proof}

The full theorem allows to unfold $Y$ in the equations for those $X$ that precede that of $Y$, and
in the equation of $Y$ itself. So in particular, the case $X=Y$ is allowed.
\begin{thm}\label{unfold}
Let $\E\in\Eqs$ be monotonic. 
Let $\S=\S_1\conc\sigma Y\conc\S_2$ and $X\notin\dom(\S_2)$.
Then $\sem{\unfold(\E,X,Y),\S}=\sem{\E,\S}$.
\end{thm}
\begin{proof}
The proof is by induction on $\S_1$. The base case is $\S_1=\varepsilon$.
If $X=Y$, we have to prove 
$\sem{\unfold(\E,X,X),\sigma X\conc\S_2}=\sem{\E,\sigma X\conc\S_2}$, 
which is just Lemma~\ref{unfoldsame}. Otherwise, if $X\neq Y$,
then $X\notin\dom(\sigma Y\conc\S_2)$, so by Lemma~\ref{unfoldaux}
$\sem{\unfold(\E,X,Y),\sigma Y\conc\S_2} = \sem{\E,\sigma Y\conc\S_2}$.

Next, for $\S=\rho Z\conc\S_1\conc\sigma Y\conc\S_2$, define
$\S_3 := \S_1\conc\sigma Y\conc\S_2$, and
assume the induction hypothesis,
$\sem{\unfold(\E,X,Y),\S_3}=\sem{\E,\S_3}$.

If $Z\neq X$, then $\E\restrictedeq{\{Z\}} \unfold(\E,X,Y)$.
From the induction hypothesis, it follows by congruence 
(Lemma~\ref{congruence})
that $\sem{\unfold(\E,X,Y),\rho Z\conc\S_3}=\sem{\E,\rho Z\conc\S_3}$.

If $Z=X$, we compute for arbitrary $\eta\in\Val$:
\allowdisplaybreaks
\begin{align*}
    & \inveq \sem{\unfold(\E,X,Y),\sigma X\conc\S_3}(\eta)\\
    & = \mbox{~~~~~(by definition of the semantics)}\\
    & \inveq \sem{\unfold(\E,X,Y),\S_3}(\eta[X:=\sigma(F)])\mbox{, where}\\
    & \inveq ~~ F(P) := \unfold(\E,X,Y)_X(\sem{\unfold(\E,X,Y),\S_3}(\eta[X:=P])) \\
    & = \mbox{~~~~~(by induction hypothesis)}\\
    & \inveq \sem{\E,\S_3}(\eta[X:=\sigma(F)])\mbox{, where}\\
    & \inveq ~~ \begin{array}[t]{rl}
    F(P) := & \unfold(\E,X,Y)_X(\sem{\E,\S_3}(\eta[X:=P])) \\
    = & \mbox{~~~~~(definition $\unfold$)}\\
      & \E_X\Big(\sem{\E,\S_3}(\eta[X:=P]) 
              \Big[Y := \E_Y(\sem{\E,\S_3}(\eta[X:=P]))\Big]\Big)\\
    = & \mbox{~~~~~(Lemma~\ref{solution}; $Y\in\dom(\S_3)$)}\\
      & \E_X\Big(\sem{\E,\S_3}(\eta[X:=P]) 
            \Big[Y := \sem{\E,\S_3}(\eta[X:=P])(Y)\Big]\Big) \\
    = & \mbox{~~~~~(congruence and extensionality: $\zeta[Y:=\zeta(Y)]=\zeta$)}\\
      & \E_X(\sem{\E,\S_3}(\eta[X:=P]))
     \end{array}\\
    & = \sem{\E,\sigma X\conc\S_3}(\eta)
\qedhere
\end{align*}
\end{proof}

\subsection{Substituting a Partial Solution}

The following theorem is motivated in~\cite{Mader96} as follows. Assume
we know by some means the solution $a$ for a variable $X$ in $(\E,\S)$.
Then we can replace the definition of $X$ by simply putting $X=a$. We
simplify the proof in~\cite{Mader96}, which is based on an infinite
series of FESs. Instead, we just use induction on $\S$ and some properties
of complete lattices.

\begin{thm}\label{partial}
Let $\E\in\Eqs$ be monotonic and let $a := \sem{\E,\S}(\eta)(X)$.
Then \[\sem{\E,\S}(\eta) = \sem{\E[X\mapsto a],\S}(\eta)\]
\end{thm}
\begin{proof}
We prove the theorem by induction on $\S$. The base case
is trivial, for 
$\sem{\E,\varepsilon}(\eta) = \eta = \sem{\E[X\mapsto a],\varepsilon}(\eta)$.
For the induction step ($\sigma Y\conc\S$), we need to define the functions
$a,b \colon \Val \to U$ and $F,G \colon U \to U$ and $H \colon (U \times U) \to U$:
\[\begin{array}{rcl}
a(\eta') & := & \sem{\E,\S}(\eta')(X)\\
b(\eta)  & := & \sem{\E,\sigma Y\conc\S}(\eta)(X)\\
F(P)     & := & \E_Y(\sem{\E,\S}(\eta[Y:=P])) \\
G(P)     & := & \E_Y(\sem{\E[X\mapsto b(\eta)],\S})\\
H(P,Q)   & := & \E_Y(\sem{\E[X\mapsto a(\eta[Y:=P])],\S}(\eta[Y:=Q]))
\end{array}\]
Then from the induction hypothesis
$ \forall \eta'.~
   \sem{\E,\S}(\eta') = \sem{\E [X\mapsto a(\eta')],\S}(\eta')$,
we must prove: 
$ \forall \eta.~ 
  \sem{\E,\sigma Y\conc\S}(\eta)=
  \sem{\E [X\mapsto b(\eta)],\sigma Y\conc\S}(\eta)$.

We distinguish three cases:

\noindent If $X=Y$ and $X\notin\dom(\S)$, 
we compute:
\begin{align*}
  & \inveq \sem{\E[X\mapsto b(\eta)],\sigma X\conc\S}(\eta) \\
  & = \sem{\E[X\mapsto b(\eta)],\S}(\eta [X:=\sigma P.\,b(\eta)])\\
  & = \mbox{~~~~~(by the constant rule, Lemma~\ref{fixcalculus}.2)}\\
  & \inveq \sem{\E[X\mapsto b(\eta)],\S}(\eta [X := b(\eta)])\\
  & = \mbox{~~~~~(by Lemma~\ref{sanity}.2 and $X\notin\dom(\S)$)}\\
  & \inveq \sem{\E,\S}(\eta [X := b(\eta)])\\
  & = \mbox{~~~~~(unfold definition of $\sem{\ }$ in $b(\eta)$.)}\\
  & \inveq \sem{\E,\S}(\eta [X := \sem{\E,\S}(\eta[X:=\sigma(F)])(X)])\\\
  & = \mbox{~~~~~(by Lemma~\ref{sanity}.1 and $X\notin\dom(\S)$)}\\
  & \inveq \sem{\E,\S}(\eta [X := \sigma(F)])\\
  & = \sem{\E,\sigma X\conc\S}(\eta)
\end{align*}

\noindent If $X=Y$ and $X\in \dom(\S)$, then note that
using Lemma~\ref{sanity}.3, for any $\E'$ and appropriate $F'$,
we have:
\[\sem{\E',\sigma X\conc\S}(\eta)=
  \sem{\E',\S}(\eta[X:=\sigma F']) = \sem{\E',\S}(\eta)\]
So in particular, $b(\eta)=a(\eta)$, and we can apply the induction
hypothesis:
Hence
\allowdisplaybreaks
\begin{align*}
    & \sem{\E[X\mapsto b(\eta)],\sigma X\conc\S}(\eta) \\
=\; & \mbox{~~~~~(by the equality above)}\\
    & \sem{\E[X\mapsto b(\eta)],\S}(\eta)\\
=\; & \sem{\E[X\mapsto a(\eta)],\S}(\eta)\\
=\; & ~~~~~\mbox{(induction hypothesis)} \\
    & \sem{\E,\S}(\eta)\\
=\; & \mbox{~~~~~(by the equality above)}\\
    & \sem{\E,\sigma X\conc\S}(\eta)
\end{align*}

\noindent Finally, if $X\neq Y$, then we can compute:
\begin{align*}
    & \sem{\E,\sigma Y\conc\S}(\eta) \\
=\; & \sem{\E,\S}(\eta [Y:=\sigma(F)])\\
=\; & ~~~~~\mbox{(induction hypothesis)} \\
    & \sem{\E [X\mapsto a(\eta[Y:=\sigma(F)])],\S}(\eta [Y:=\sigma(F)])\\
=\; & \mbox{~~~~~(unfold definition of $\sem{\ }$ in $b(\eta)$.)}\\
    & \sem{\E [X\mapsto b(\eta)],\S}(\eta [Y:=\sigma(F)])\\
=\; & ~~~~~\mbox{(will be proved below)}\\
    & \sem{\E [X\mapsto b(\eta)],\S}(\eta [Y:=\sigma(G)])\\
=\; & \sem{\E [X\mapsto b(\eta)],\sigma Y\conc\S}(\eta)
\end{align*}
We must still prove that $\sigma(F)=\sigma(G)$.
\begin{align*}
    &\sigma P.\,F(P) \\
=\; & ~~~~~\mbox{(induction hypothesis)} \\
    & \sigma P.\,\E_Y(\sem{\E[X\mapsto a(\eta[Y:=P])],\S}(\eta[Y:=P]))\\
=\; & \sigma P.\,H(P,P)\\
=\; & ~~~~~\mbox{Lemma~\ref{fixcalculus}.8 (solve rule) on $\lambda P,Q.\,H(Q,P)$} \\
    & \sigma P.\,H(\sigma P.\,H(P,P),P)\\
=\; & ~~~~~\mbox{($F(P)=H(P,P)$ as above)}\\
    & \sigma P.\,H(\sigma P.\,F(P),P)\\
=\; & ~~~~~\mbox{(unfold definition of $\sem{~}$ in $b$ in $G$)}\\
    & \sigma P.\,G(P)
\qedhere
\end{align*}
\end{proof}

\section{Swapping Variables in FES specifications}\label{swap}
We now study swapping the order of variables in a specification.  In
general, this operation doesn't exactly preserve solutions. We first
show that adjacent variables with the same sign may be swapped without
changing the semantics (Section~\ref{samesign}). Subsequently, we will prove
that we can swap the order between blocks of equations, under certain
independence criteria (Section~\ref{independent}). The main theorem of
that section is new (cf. Section~\ref{related}). Finally, we show
that swapping a $\mu/\nu$ sequence by the corresponding $\nu/\mu$ in
general leads to a greater or equal solution (Section~\ref{diffsign}).

\subsection{Swapping Equations with the same Sign}
\label{samesign}
\begin{thm}\label{swapsame}
Assume $\E\in\Eqs$ is monotonic. Then
\[\sem{\E,\S_1\conc\sigma X\conc\sigma Y\conc\S_2}=
 \sem{\E,\S_1\conc\sigma Y\conc\sigma X\conc\S_2}\]
\end{thm}
\begin{proof}
We first compute for arbitrary $\eta\in\Val$:
\allowdisplaybreaks
\begin{align*}
     & \sem{\E,\sigma X\conc\sigma Y\conc\S_2}(\eta) \\
 =\; & \sem{\E,\sigma Y\conc\S_2}(\eta[X := \sigma(F_2)])\mbox{, where}\\
     &\qquad\begin{aligned}
        F_2(P) &= \E_X(\sem{\E,\sigma Y\conc\S_2}(\eta[X:= P]))\\
      \end{aligned}\\
 =\; & \sem{\E,\S_2}(\eta[X:=\sigma(F_2),\, Y:=\sigma(F_3)])\mbox{, where}\\
     &\qquad\begin{aligned}
        F_2(P) &= \E_X(\sem{\E,\sigma Y\conc\S_2}(\eta[X:= P]))\\
        F_3(Q) &= \E_Y(\sem{\E,\S_2}(\eta[X:=\sigma(F_2),\, Y:=Q]))\\
      \end{aligned}\\
 =\; & \mbox{~~~~~(unfold definition of $\sem{~}$ in $F_2$, introduce $F_1$)}\\
     & \sem{\E,\S_2}(\eta[X:=\sigma(F_2),\, Y:=\sigma(F_3)])\mbox{, where}\\
     &\qquad\begin{aligned}
         F_1(P)(Q) &= \E_Y(\sem{\E,\S_2}(\eta[X:= P,\, Y:=Q]))\\
         F_2(P)    &= \E_X(\sem{\E,\S_2}(\eta[X:= P,\, Y:=\sigma(F_1(P))]))\\
         F_3(Q)    &= \E_Y(\sem{\E,\S_2}(\eta[X:=\sigma(F_2),\, Y:=Q]))\\
      \end{aligned}\\
 =\; & A(\sigma(F_2),\sigma(F_3))\mbox{, where}\\
     &\qquad\begin{aligned}
         A(P,Q)    &= \sem{\E,\S_2}(\eta[X:= P,\, Y:=Q]) \\
         F_1(P)(Q) &= \E_Y(A(P,Q)) \\
         F_2(P)    &= \E_X(A(P,\sigma(F_1(P))))\\
         F_3(Q)    &= \E_Y(A(\sigma(F_2),Q))\\
      \end{aligned}
\end{align*}
Symmetrically, we get:
\begin{align*}
     & \sem{\E,\sigma Y\conc\sigma X\conc\S_2}(\eta) \\
 =\; & B(\sigma(G_2),\sigma(G_3))\mbox{, where}\\
     &\qquad\begin{aligned}
        B(Q,P)    &= \sem{\E,\S_2}(\eta[Y:= Q,\, X:=P]) \\
        G_1(Q)(P) &= \E_X(B(Q,P)) \\
        G_2(Q)    &= \E_Y(B(Q,\sigma(G_1(Q))))\\
        G_3(P)    &= \E_X(B(\sigma(G_2),P))\\
      \end{aligned}
\end{align*}
Note that the theorem is trivial when $X=Y$.
So we may assume $X\neq Y$. Hence $A(P,Q)=B(Q,P)$, and we have:
\begin{align*}
    & \sigma     (F_2)\\
=\; & \sigma P.\,\E_X(A(P,\sigma(F_1(P)         ))) \\
=\; & \sigma P.\,\E_X(A(P,\sigma Q.\,\E_Y(A(P,Q)))) \\
=\; & \begin{array}[t]{l}
	  \mbox{~~~~~(Beki\v{c} rule, Lemma~\ref{fixcalculus}.9, 
                  with $H(p,q) := \E_X(A(p,q))$ }\\
      \quad\mbox{~~~~~~and $K(p,q):=\E_Y(A(q,p))$, which are monotonic, because~}\\
      \quad\mbox{~~~~~~$\E$ is by assumption, and 
                  $\sem{\E,\S_2}$ by Lemma~\ref{monotonic}.1)}
      \end{array}\\
    & \sigma P.\,\E_X(A(P,\sigma Q.\,\E_Y(A(\sigma P.\,\E_X(A(P,Q)),Q))))\\
=\; & \mbox{~~~~~(while $A(P,Q)=B(Q,P)$~)}\\
    & \sigma P.\,\E_X(B(\sigma Q.\,\E_Y(B(Q,\sigma P.\,\E_X(B(Q,P))))),P)\\
=\; & \sigma P.\,\E_X(B(\sigma Q.\,\E_Y(B(Q,\sigma     (G_1(Q))))     ,P))\\
=\; & \sigma P.\,\E_X(B(\sigma     (G_2)                              ,P))\\
=\; & \sigma     (G_3)
\end{align*}

We can now finish the proof:
\begin{align*}
             & \mbox{~~~~~(computation above, and full symmetry)}\\
             & \sigma(F_2)=\sigma(G_3) \mbox{~~and~~} \sigma(F_3)=\sigma(G_2)\\
\Rightarrow\;& \mbox{~~~~~(because $A(P,Q)=B(Q,P)$~)}\\
             & A(\sigma(F_2),\sigma(F_3)) = B(\sigma(G_2),\sigma(G_3)) \\
\Rightarrow\;& \sem{\E,\sigma X\conc\sigma Y\conc\S_2}(\eta) =
               \sem{\E,\sigma Y\conc\sigma X\conc\S_2}(\eta)\\
\Rightarrow\;& \mbox{~~~~~(Lemma~\ref{congruence}})\\
             & \sem{\E,\S_1\conc\sigma X\conc\sigma Y\conc\S_2}(\eta) =
               \sem{\E,\S_1\conc\sigma Y\conc\sigma X\conc\S_2}(\eta)
\qedhere
\end{align*}
\end{proof}

\subsection{Migrating Independent Blocks of Equations}
\label{independent}

Our aim here is to investigate swapping blocks of equations that
are independent. We first need two technical lemmas. 
The first lemma enables to commute updates to valuations with computing
solutions:
\begin{lem}\label{indepenv}
If $X\notin\dom(\S)$ and $\indep(\E,\S,X)$, then
\[\sem{\E,\S}(\eta[X:=P]) = (\sem{\E,\S}(\eta))[X:=P]\]
\end{lem}
\begin{proof}
Induction on $\S$. The base case is trivial:
\[ \sem{\E,\varepsilon}(\eta[X:=P])
 = \eta[X:=P]
 = \sem{\E,\varepsilon}(\eta) [X:=P]\]
Case $\S=\sigma Y\conc\S_1$. Assume $X\notin\dom(\S)$ and $\indep(\E,\S, X)$,
then $X\neq Y$, and also $X\notin\dom(\S_1)$ and $\indep(\E,\S_1, X)$, so the
induction hypothesis can be applied. Then
\allowdisplaybreaks
\begin{align*}
     & \sem{\E,\sigma Y\conc\S_1}(\eta[X:= P])\\
 =\; & \sem{\E,\S_1}(\eta[X:=P,\,Y:=\sigma(F)])\mbox{, where}\\
       &\qquad\begin{aligned}
           F(Q) :=\: &\; \E_Y(\sem{\E,\S_1}(\eta[X:= P,\, Y:=Q])) \\
                \:=\:&\;  ~~~~~\mbox{(induction hypothesis, and $X\neq Y$)} \\
                     &\; \E_Y(\sem{\E,\S_1}(\eta[Y:=Q])[X:=P])\\ 
                \:=\:&\;  ~~~~~\mbox{($Y\in\dom(\S)$ is independent of $X$)}\\
                     &\; \E_Y(\sem{\E,\S_1}(\eta[Y:=Q]))\\
                \:=: &\; G(Q)\\
       \end{aligned}\\
 =\; & \sem{\E,\S_1}(\eta[X:=P,\,Y:=\sigma(G)])\\
 =\; & ~~~~~\mbox{(induction hypothesis, and $X \neq Y$)} \\
     & \sem{\E,\S_1}(\eta[Y:=\sigma(G)])[X:=P]\\
 =\; & \sem{\E,\sigma Y\conc\S_1}(\eta)[X:=P]
\qedhere
\end{align*}
\end{proof}

The next lemma states that independent specifications can be solved
independently.
\begin{lem}\label{indepsolve}
Let $\E\in\Eqs$ and $\S_1,\S_2\in\Spec$. If
$\indep(\E,\S_1,\S_2)$, then for all $\eta\in\Val$,
$\sem{\E,\S_1\conc\S_2}(\eta) = \sem{\E,\S_2}(\sem{\E,\S_1}(\eta))$.
\end{lem}
\begin{proof}
Induction on $\S_1$. In case $\S_1=\varepsilon$ we obtain indeed:
\[\sem{\E,\varepsilon\conc\S_2}(\eta)
 =\sem{\E,\S_2}(\eta) 
 = \sem{\E,\S_2}(\sem{\E,\varepsilon}(\eta))\]
Next, consider $\S_1=\sigma X\conc\S$. Assume $\indep(\E,\S_1,\S_2)$,
then it follows that $\indep(\E,\S,\S_2)$,
so we can use the induction hypothesis. Define:
\[\begin{array}{rcl}
      F(P) & := & \E_X(\sem{\E,\S\conc\S_2}(\eta[X:= P]))\\
      G(P) & := & \E_X(\sem{\E,\S}(\eta[X:= P]))\\
\end{array}\]

In order to show that $F=G$, it suffices (because $\E_X$ is 
{\em independent} of $\S_2$) to show that for any $P\in U$
and $Y\notin\dom(\S_2)$:
\begin{align*}
    & \sem{\E,\S\conc\S_2}(\eta[X:= P])(Y)\\
 =\;& ~~~~~\mbox{(induction hypothesis)} \\
    & \sem{\E,\S_2}(\sem{\E,\S}(\eta[X:=P]))(Y)\\ 
 =\;& ~~~~~\mbox{(Lemma~\ref{sanity}.1)} \\
    & \sem{\E,\S}(\eta[X:=P])(Y)
\end{align*}
Next, we finish the proof with the following calculation:
\begin{align*}
    & \sem{\E,\sigma X\conc\S\conc\S_2}(\eta)  \\
=\; & \sem{\E,\S\conc\S_2}(\eta [X:=\sigma(F)])\\
=\; & ~~~~~\mbox{(induction hypothesis)} \\
    & \sem{\E,\S_2}(\sem{\E,\S}(\eta[X:=\sigma(F)]))\\ 
=\; & ~~~~~\mbox{($F=G$, see above)} \\
    & \sem{\E,\S_2}(\sem{\E,\S}(\eta[X:=\sigma(G)]))\\ 
=\; & \sem{\E,\S_2}(\sem{\E,\sigma X\conc\S}(\eta))
\qedhere
\end{align*}
\end{proof}

The next theorem shows that two disjoint blocks of equations can be
swapped, provided one of them doesn't depend on the other.  Note that
a dependence in one direction is allowed, and that it doesn't matter
in which direction by symmetry. Theorem~\ref{migration2} will
generalize this by adding left- and right-contexts under certain
conditions.

\begin{thm}\label{migration}
Let $\disjoint(\S_1,\S_2)$ and $\indep(\E,\S_1,\S_2)$.
Then $\sem{\E,\S_1\conc\S_2} = \sem{\E,\S_2\conc\S_1}$.
\end{thm}
\begin{proof}
Induction on $\S_2$.
The base case is trivial.
For the induction step, let $\S_2 = \sigma X\conc\S$.
If we assume $\disjoint(\S_1,\S_2)$ and $\indep(\E,\S_1,\S_2)$, then we 
also obtain $\disjoint(\S_1,\S)$ and $\indep(\E,\S_1,\S)$, so we may 
apply the induction hypothesis $\sem{\E,\S_1\conc\S} = 
\sem{\E,\S\conc\S_1}$.
Furthermore, from the same assumptions, we also
get $X \notin \dom(\S_1)$ and $\indep(\E,\S_1,X)$.
Let $\eta$ be arbitrary.
\allowdisplaybreaks
\begin{align*}
    & \sem{\E,\S_1\conc\sigma X\conc\S}(\eta)  \\
=\; & \mbox{~~~~~(Lemma~\ref{indepsolve})}\\
    & \sem{\E,\sigma X\conc\S}(\sem{\E,\S_1}(\eta))\\
=\; & \sem{\E,\S}(\sem{\E,\S_1}(\eta)[X:=\sigma(F)])\mbox{, where}\\
        &\qquad\begin{aligned}
          F(P) :=\: &\; \E_X(\sem{\E,\S}(\sem{\E,\S_1}(\eta)[X:= P]) \\
                =\: &\; \mbox{~~~~~(Lemma~\ref{indepenv})}\\
                    &\; \E_X(\sem{\E,\S}(\sem{\E,\S_1}(\eta[X:=P])))\\
                =\: &\; \mbox{~~~~~(Lemma~\ref{indepsolve})}\\
                    &\; \E_X(\sem{\E,\S_1\conc\S}(\eta[X:=P]))\\
                =:  &\; G(P)\\
        \end{aligned}\\
=\; & \sem{\E,\S}(\sem{\E,\S_1}(\eta)[X:=\sigma(G)])\\
=\; & \mbox{~~~~~(Lemma~\ref{indepenv})}\\
    & \sem{\E,\S}(\sem{\E,\S_1}(\eta[X:=\sigma(G)]))\\
=\; & \mbox{~~~~~(Lemma~\ref{indepsolve})}\\
    & \sem{\E,\S_1\conc\S}(\eta[X:=\sigma(G)])\\
=\; & \sem{\E,\sigma X\conc\S_1\conc\S}(\eta)\\
=\; & \mbox{~~~~~(Lemma~\ref{congruence} and induction hypothesis)}\\
=\; & \sem{\E,\sigma X\conc\S\conc\S_1}(\eta)
\qedhere
\end{align*}
\end{proof}

This migration theorem has several interesting corollaries.
First, we get right-congruence for independent specifications.
\begin{cor}\label{congright}
Assume that $\indep(\E_1,\S_1,\S)$, $\indep(\E_2,\S_2,\S)$,
$\disjoint(\S,\S_1\conc\S_2)$ and 
that $\E_1 \restrictedeq{\S} \E_2$.
Then $\sem{\E_1,\S_1}=\sem{\E_2,\S_2}$ implies
$\sem{\E_1,\S_1\conc\S}=\sem{\E_2,\S_2\conc\S}$.
\end{cor}

We also get the near-reverse of Lemma~\ref{indepsolve}:
\begin{cor}\label{indepsolve2}
Let $\indep(\E,\S_2,\S_1)$ and $\disjoint(\S_1,\S_2)$. 
Then for all $\eta\in\Val$, we have
$\sem{\E,\S_1\conc\S_2}(\eta) = \sem{\E,\S_1}(\sem{\E,\S_2}(\eta))$.
\end{cor}
\begin{proof}
Under the given assumptions, we obtain from Theorem~\ref{migration}
(applied from right to left) and Lemma~\ref{indepsolve}:
\[\sem{\E,\S_1\conc\S_2}(\eta) = 
  \sem{\E,\S_2\conc\S_1}(\eta) = 
  \sem{\E,\S_1}(\sem{\E,\S_2}(\eta))\qedhere\]
\end{proof}

\begin{thm}\label{migration2}
Assume that $\disjoint(\S_1,\S_2\conc\S_3)$.
Also, assume that either $\indep(\E,\S_1,\S_2\conc\S_3)$, or 
$\indep(\E,\S_2\conc\S_3,\S_1)$.
Then 
${\sem{\E,\S_0\conc\S_1\conc\S_2\conc\S_3} =
           \sem{\E,\S_0\conc\S_2\conc\S_1\conc\S_3}}$.
\end{thm}
\begin{proof}
We also have $\disjoint(\S_1,\S_3)$, and either
$\indep(\E,\S_1,\S_3)$ or $\indep(\E,\S_3,\S_1)$.
Using Lemma~\ref{congruence} and Theorem~\ref{migration} twice
(in either direction) we get:
\[ \sem{\E,\S_0\conc\S_1\conc\S_2\conc\S_3}
 = \sem{\E,\S_0\conc\S_2\conc\S_3\conc\S_1}
 = \sem{\E,\S_0\conc\S_2\conc\S_1\conc\S_3}\qedhere
\]
\end{proof}

\subsection{Inequalities by Swapping or Changing Signs}
\label{diffsign}

In this section, we will prove a few inequalities. Theorem~\ref{migineq}
shows the consequence of swapping variables with a different sign; 
Theorem~\ref{signineq} shows the effect of changing the sign of a variable.
But first, it will be shown that $\leq$ is a left congruence:

\begin{lem}\label{left_ineq_cong}
Let $\E_1,\E_2\in\Eqs$ and $\S,\S_1,\S_2\in\Spec$.
If $\E_1$ is monotonic, $\E_1 \restrictedeq{\S} \E_2$, and
$\sem{\E_1,\S_1}\leq\sem{\E_2,\S_2}$, then
$\sem{\E_1,\S\conc\S_1}\leq\sem{\E_2,\S\conc\S_2}$.
\end{lem}
\begin{proof}
Induction on $\S$. The base case is trivial.
\allowdisplaybreaks
\begin{align*}
       & \sem{\E_1,\sigma X\conc\S\conc\S_1}(\eta)\\
=\;    & \sem{\E_1,\S\conc\S_1}(\eta [X:=\sigma(F)]), where\\
       &\qquad\begin{aligned}
          F(P) :=\: &\; \E_{1,X}(\sem{\E_1,\S\conc\S_1}(\eta [X:=P]))\\
              \leq\:&\; \mbox{~~~~~(by induction hypothesis and $\E_1$ monotonic)}\\
                    &\; \E_{1,X}(\sem{\E_2,\S\conc\S_2}(\eta [X:=P]))\\
               =\:  &\; \E_{2,X}(\sem{\E_2,\S\conc\S_2}(\eta [X:=P]))\\
                =:  &\; G(P)
       \end{aligned}\\
\leq\; & \mbox{~~~~~(Using Lemma~\ref{monotonic}.1)}\\
       &    \sem{\E_1,\S\conc\S_1}(\eta [X:=\sigma(G)]) \\
\leq\; & \mbox{~~~~~(by induction hypothesis)}\\
       &       \sem{\E_2,\S\conc\S_2}(\eta [X:=\sigma(G)]) \\
=\;    & \sem{\E_2,\sigma X\conc\S\conc\S_2}(\eta)
\qedhere
\end{align*}
\end{proof}
Note that, by duality, the above lemma may also be applied if $\E_2$
is monotonic instead of $\E_1$.
Moreover, note that (only) for monotonic $\E_1$, Lemma~\ref{congruence} 
would follow from Lemma~\ref{left_ineq_cong}.

\begin{thm}\label{migineq}
Assume $\E\in\Eqs$ is monotonic and $X\neq Y$. Then
\[\sem{\E,\S_1\conc\mu X\conc\nu Y\conc\S_2}\leq
 \sem{\E,\S_1\conc\nu Y\conc\mu X\conc\S_2}\]
\end{thm}
\begin{proof}
As in Theorem~\ref{swapsame}, and using $X\neq Y$, we obtain:
\begin{align*}
   \sem{\E,\mu X\conc\nu Y\conc\S_2}(\eta)
&= A(\mu(F_2),\nu(F_3))\mbox{, where}\\
  & \qquad\begin{aligned}
      A(P,Q)    &= \sem{\E,\S_2}(\eta[X:= P,\, Y:=Q]) \\
      F_1(P)(Q) &= \E_Y(A(P,Q)) \\
      F_2(P)    &= \E_X(A(P,\nu(F_1(P))))\\
      F_3(Q)    &= \E_Y(A(\mu(F_2),Q))\\
    \end{aligned}
\end{align*}
\begin{align*}
   \sem{\E,\nu Y\conc\mu X\conc\S_2}(\eta)
 &= A(\mu(G_3),\nu(G_2))\mbox{, where}\\
   &\qquad\begin{aligned}
       G_1(Q)(P) &= \E_X(A(P,Q)) \\
       G_2(Q)    &= \E_Y(A(\mu(G_1(Q)),Q))\\
       G_3(P)    &= \E_X(A(P,\nu(G_2)))\\
    \end{aligned}
\end{align*}
By Lemma~\ref{fixinequalities}.4(a), $\mu(F_2)\leq\mu(G_3)$, and
by Lemma~\ref{fixinequalities}.4(b), $\nu(F_3)\leq\nu(G_2)$, whence
it follows that $\sem{\E,\mu X\conc\nu Y\conc\S_2}(\eta)\leq
 \sem{\E,\nu Y\conc\mu X\conc\S_2}(\eta)$. The theorem then follows
by Lemma~\ref{left_ineq_cong}.
\end{proof}

We end this section with another inequality:
\begin{thm}\label{signineq}
If $\E$ is monotonic, then 
$\sem{\E,\S_1\conc\mu X\conc\S_2} \leq \sem{\E,\S_1\conc\nu X\conc\S_2}$.
\end{thm}
\begin{proof}
Let $\eta$ be an arbitrary valuation, and define $F:U\to U$ by\\
$F(P) := \E_X(\sem{\E,\S_2}(\eta[X:=P]))$. We then have:
\begin{align*}
             & \mbox{~~~~(Theorem~\ref{fixinequalities}.1)}\\
             & \mu P.\,F(P) \leq \nu P.\,F(P)\\
\Rightarrow\;& \mbox{~~~~(Monotonicity, Theorem~\ref{monotonic}.1)}\\
             & \sem{\E,\S_2}(\eta[X := \mu P.\,F(P)]) \leq
               \sem{\E,\S_2}(\eta[X := \nu P.\,F(P)])\\
\Rightarrow\;& \mbox{~~~~(Definition semantics)}\\
             & \sem{\E,\mu X\conc\S}(\eta) \leq \sem{\E,\nu X\conc\S_2}(\eta)\\
\Rightarrow\;& \mbox{~~~~($\eta$ was arbitrary)}\\
             & \sem{\E,\mu X\conc\S} \leq \sem{\E,\nu X\conc\S_2}\\
\Rightarrow\;& \mbox{~~~~(Congruence, Theorem~\ref{left_ineq_cong})}\\
             & \sem{\E,\S_1\conc\mu X\conc\S_2} \leq \sem{\E,\S_1\conc\nu X\conc\S_2}
\qedhere
\end{align*}
\end{proof}

In the next section, we will see sufficient conditions under which the
inequality signs of these theorems can be turned into
equalities. These conditions will be phrased in terms of the
dependency graph.

\section{Indirect Dependencies and Loops}\label{sec:independent}
In Section~\ref{independent}, we studied direct dependencies between
variables. Basically, a direct dependency of $X$ on $Y$ means
that $Y$ occurs in the definition of $X$. We will now study the
effect of indirect dependencies, written $X\dep{\E,\S}Y$
(cf.~the definitions in Section~\ref{depgraph})

Given a specification $\S$ and a computable predicate $f$ on variables, we define $\splitname_f(\S) = (\S_1,\S_2)$, where $\S_1$ is the sublist of $\S$ with those $X$ for which $f(X)$ does \emph{not} hold and $\S_2$ is the sublist of $\S$ with those $X$ for which $f(X)$ \emph{does} hold.
Notice that, within $\S_1$ and $\S_2$, variables keep their order from $\S$.

The following basic facts follow directly from the definition of $\splitname$.
\begin{lem}\label{splitfacts1}
Let $\splitname_f(\S)=(\S_1,\S_2)$, then $\dom(\S)=\dom(\S_1)\cup\dom(\S_2)$ and\linebreak$\disjoint(\S_1,\S_2)$.
\end{lem}

We first show how the equations in a FES may be rearranged if the specification is
split in such a way that certain independence conditions are fulfilled.

\begin{lem}\label{splitsolve_aux}
Let $f$ be a predicate and $\S$ a specification such that $\splitname_f(\S) = (\S_1, \S_2)$ and $\indep(\E,\S_2,\S_1)$.
Then $\sem{\E, \S} = \sem{\E, \S_1\conc\S_2}$.
\end{lem}
\begin{proof}
We perform induction on $\S$.
The base case is trivial.
Let $\splitname_f(\S) = (\S_1,\S_2)$ and assume as induction hypothesis that $\indep(\E,\S_2,\S_1)$ implies $\sem{\E,\S} = \sem{\E,\S_1\conc\S_2}$.
For the induction step, we consider the specification $\sigma Y\conc\S$ and distinguish two cases.

If $f(Y)$ does not hold, then we have $\splitname_f(\sigma Y\conc\S) = (\sigma Y\conc\S_1,\S_2)$.
Accordingly, we assume $\indep(\E,\sigma Y \conc \S_2,\S_1)$, which implies $\indep(\E,\S_2,\S_1)$.
We can thus apply the induction hypothesis and congruence (Lemma~\ref{congruence}) to obtain
$\sem{\E,\sigma Y\conc\S} = \sem{\E,\sigma Y\conc\S_1\conc\S_2}$.

Otherwise, if $f(Y)$ holds, we obtain $\splitname_f(\sigma Y\conc\S) = (\S_1,\sigma Y\conc\S_2)$.
Now we assume that $\indep(\E,\S_2,\sigma Y \conc \S_1)$, which again implies $\indep(\E,\S_2,\S_1)$. 
We also have $\disjoint(\S_1,\sigma Y \conc \S_2)$ (Lemma~\ref{splitfacts1}), so we may apply Theorem~\ref{migration2} below:
\begin{align*}
   & \sem{\E,\sigma Y \conc\S}\\
=\;& \mbox{~~~~~(by induction hypothesis and Lemma~\ref{congruence})}\\
   & \sem{\E,\sigma Y \conc\S_1\conc\S_2}\\
=\;& \mbox{~~~~~(by Theorem~\ref{migration2})}\\
   & \sem{\E,\S_1\conc\sigma Y \conc\S_2}
\qedhere
\end{align*}
\end{proof}

We simplify notation and write $\splitS{X}{\E}{\S}$ for 
$\splitname_{\mathit{dep}^X_{\E,\S}}(\S)$, defining the predicate
$\smash{\mathit{dep}^X_{\E,\S}}(Y) = X \dep{\E,\S} Y$.
If $\indep(\E,X,Y)$ is computable (which we assume henceforward), then $X \dep{\E,\S} Y$ is also computable since $\S$ is finite.
Intuitively, if $\splitS{X}{\E}{\S}=(\S_1,\S_2)$, then
$\S_1$ is the sublist of $\S$ with those $Y$ on which $X$ 
\emph{does not} depend indirectly; and $\S_2$
is the sublist of $\S$ with those $Z$ on which $X$ \emph{does} depend 
indirectly.
Notice that, if $X \notin \dom(\S)$, then $\splitS{X}{\E}{\S} = (\S, \varepsilon)$.

We have the following lemma about splitting a specification based on the dependencies of $X$:

\begin{lem}\label{splitfacts2}
	If $\splitS{X}{\E}{\S}=(\S_1,\S_2)$, then $\indep(\E,\S_2,\S_1)$.
\end{lem}
\begin{proof}
Assume that some $Z\in\dom(\S_2)$
would use some $Y\in\dom(\S_1)$ in its definition in $\E$. 
Then $X\dep{\E,\S} Z \depstep{\E,\S} Y$, so $X\dep{\E,\S}Y$,
and $Y$ would be in $\S_2$ and not in $\S_1$.
\end{proof}

The key theorem of this section states that the equations in a FES can
be rearranged, such that all equations that $X$ depends on precede all
other equations, or vice versa.  This is useful, because those parts
can be solved independently, using Lemma~\ref{indepsolve}. By
repeatedly picking a variable in a terminal strongly connected component
of the remaining variable dependency graph, one can thus solve all
SCCs one by one.
This idea already appeared in~\cite{jurdzinski_small_2000} for parity games, although it
does not always provide performance benefits in practice~\cite{friedmann_solving_2009}.
The theorem may also be used to reduce the number of fixpoint alternations
in a FES.

\begin{thm}\label{splitsolve}
Let $\splitS{X}{\E}{\S}=(\S_1,\S_2)$. Then 
$\sem{\E,\S} = \sem{\E,\S_1\conc\S_2} = \sem{\E,\S_2\conc\S_1}$.
\end{thm}
\begin{proof}
The first equality follows from Lemmas~\ref{splitsolve_aux} and~\ref{splitfacts2}.
From Theorem~\ref{migration} and Lemmas~\ref{splitfacts1} and~\ref{splitfacts2} it
follows that also $\sem{\E,\S_1\conc\S_2} = \sem{\E,\S_2\conc\S_1}$.
\end{proof}

Based on this reordering principle, we can prove three more interesting
results, which we will do in the next subsections. 

\subsection{Swapping Signs and Dependency Loops}
The first result (Theorem~\ref{signloop}) states that the sign of a variable
$X$ is only relevant if it depends on itself, \emph{i.e.}, $X$ is on a cycle in $\depnonempty{\E,\S}$ (recall that $\depnonempty{\E,\S}$ indicates a
non-empty path in the variable dependency graph).
We first need a couple of auxiliary lemmas:

\begin{lem}\label{signindep}
If $X\notin\dom(\S)$ and we have $\indep(\E,\S,X)$ as well as $\indep(\E,X,X)$,
then $\sem{\E,\mu X\conc\S} = \sem{\E,\nu X\conc S}$.
\end{lem}
\begin{proof}
For $\sigma\in\{\mu,\nu\}$ and arbitrary valuation $\eta$, we have:
\allowdisplaybreaks
\begin{align*}
  &\sem{\E,\sigma X\conc\S}(\eta)\\
=\; &\mbox{~~~~~(by definition of semantics)}\\
  & \sem{\E,\S}(\eta[X:=\sigma(F)])\mbox{, where}\\
  & \qquad\begin{aligned}
      F(P) := &\; \E_X(\sem{\E,\S}(\eta[X:=P])) \\
            = &\; \mbox{~~~~~(Lemma~\ref{indepenv}, and $X\notin\dom(\S)$ and $\indep(\E,\S,X)$}\\
              &\; \E_X((\sem{\E,\S}\eta)[X:=P])\\
            = &\; \mbox{~~~~~(by definition of $\indep(\E,X,X)$)}\\
              &\; \E_X(\sem{\E,\S}\eta)\\
    \end{aligned}\\
=\; & \mbox{~~~~~(constant rule, Lemma~\ref{fixcalculus}.2)}\\
  & \sem{\E,\S}(\eta[X := \E_X(\sem{\E,\S}\eta)])
\end{align*}
So indeed $\sem{\E,\mu X\conc\S}=\sem{\E,\nu X\conc\S}$.
\end{proof}

\begin{lem}\label{sign}
If not $X\depnonempty{\E,\mu X\conc\S}X$, and $X\notin\dom(\S)$, then
$\sem{\E,\mu X\conc\S} = \sem{\E,\nu X\conc\S}$.
\end{lem}
\begin{proof}
Let $\S_1$ and $\S_2$ be such that $\splitS{X}{\E}{\sigma X\conc\S} = 
(\S_1,\sigma X\conc\S_2)$, for $\sigma\in\{\mu,\nu\}$.
Note that if not $\indep(\E,\mu X\conc\S_2,X)$, then for some
$Y \in \dom(\mu X\conc\S_2)$, by definition of $\splitname$, 
$X\dep{\E,\mu X\conc\S}Y\depstep{\E}X$, which contradicts the assumption not 
$X\depnonempty{\E,\mu X\conc\S}X$.
From $\dom(\S_2) \subseteq \dom(\S)$, we obtain $X \notin \dom(\S_2)$.
Hence, $X\notin\dom(\S_2)$ and $\indep(\E,\mu X\conc\S_2,X)$, so 
Lemma~\ref{signindep} applies. Together with Theorem~\ref{splitsolve} and 
Lemma~\ref{congruence},
we then compute:
\[ \sem{\E,\mu X\conc\S} =
   \sem{\E,\S_1\conc\mu X\conc\S_2} =
   \sem{\E,\S_1\conc\nu X\conc\S_2} =
   \sem{\E,\nu X\conc\S}
\qedhere
\]
\end{proof}

Intuitively, the sign of $X$ is only relevant if $X$ is the most
relevant variable (i.e. leftmost in the specification) on some loop in
the dependency graph.
So in the full theorem, we can restrict to
dependencies through variables right from $X$:

\begin{thm}\label{signloop}
Assume that not $X\depnonempty{\E,\mu X\conc\S_2}X$, and 
$X\notin\dom(\S_2)$.
Then \[\sem{\E,\S_1\conc\mu X\conc\S_2} = \sem{\E,\S_1\conc\nu X\conc\S_2}\]
\end{thm}
\begin{proof}
By Lemma~\ref{sign}, 
$\sem{\E,\mu X\conc\S_2} = \sem{\E,\nu X\conc\S_2}$.
The result follows by congruence,\linebreak Lemma~\ref{congruence}.
\end{proof}

\subsection{Reordering Variables and Dependency Loops}

The second result (Theorem~\ref{swaploop}) allows to swap any two neighbouring variables
that don't occur on a loop in the dependency graph.

\begin{lem}\label{swaplemma}
Let not $X\dep{\E,\sigma X\conc\rho Y\conc\S} Y$.
Then $\sem{\E,\sigma X\conc\rho Y\conc\S} 
= \sem{\E,\rho Y\conc\sigma X\conc\S}$.
\end{lem}
\begin{proof}
Note that for some $\S_1$ and $\S_2$, we have
\[ \splitS{X}{\E}{\sigma X\conc\rho Y\conc\S}=(\rho Y\conc\S_1,\sigma X\conc\S_2)
=\splitS{X}{\E}{\rho Y\conc \sigma X\conc\S}\]
Hence, by applying Theorem~\ref{splitsolve} twice, we obtain:
\[  \sem{\E,\sigma X\conc\rho Y\conc\S} 
  = \sem{\E,\rho Y\conc\S_1\conc\sigma X\conc\S_2} 
  = \sem{\E,\rho Y\conc\sigma X\conc\S}
\qedhere
\]
\end{proof}

Again, we can strengthen this, by observing that $X$ and $Y$ can be swapped,
when there is no loop that has either $X$ or $Y$ as its most relevant
variable in the specification:

\begin{thm}\label{swaploop}
Assume that not $X\dep{\E, \sigma X\conc\rho Y\conc\S_2} Y$. 
Then we have
\[
	\sem{\E,\S_1\conc\sigma X\conc\rho Y\conc\S_2} 
    = \sem{\E,\S_1\conc\rho Y\conc\sigma X\conc\S_2}
\]
\end{thm}
\begin{proof}
By Lemma~\ref{swaplemma},
$\sem{\E,\sigma X\conc\rho Y\conc\S_2} = \sem{\E,\rho Y\conc\sigma X\conc\S_2}$.
The result then follows by congruence, Lemma~\ref{congruence}.
\end{proof}

Note that this result strengthens Theorem~\ref{swapsame} (the signs
may now be different), Theorem~\ref{migration2} (we can have mutual
dependencies on $\S$, as long as no loop is introduced)
and Theorem~\ref{migineq} (we have here equality rather than inequality).

\subsection{Forward Substitution and Dependency Loops}
The final result strengthens Theorem~\ref{unfold} by
allowing unfolding of $Y$ in the definition of $X$, even if
$Y$ precedes $X$ in the specification, provided $Y$ doesn't depend on $X$:

\begin{thm}\label{unfoldloop}
Let $\E\in\Eqs$ be monotonic. Let $\S=\S_1,\sigma Y,\S_2$. 
Assume that not $Y\dep{\E,\S}X$.
Then  $\sem{\E,\S}=\sem{\unfold(\E,X,Y),\S}$.
\end{thm}
\begin{proof}
Assume not $Y\dep{\E,\S}X$. Then the following two observations hold:
\begin{enumerate}
\item for all $Z\in\X$, $Y\dep{\E,\S}Z \iff Y\dep{\unfold(\E,X,Y),\S}Z$
\item $\splitS{Y}{\E}{\S}=\splitS{Y}{\unfold(\E,X,Y)}{\S}$
\end{enumerate}
The first item holds, because $\unfold(\E,X,Y)$ only modifies the definition
of $X$, but $Y$ doesn't refer to it. The second then follows from the
definition of $\splitname$.

Let $(L_1,L_2):=\splitS{Y}{\E}{\S}$. Then, as $Y\dep{\E,\S} Y$, we have
$L_2=L_3\conc\sigma Y\conc L_4$. Note that $X\notin\dom(L_4)$, for we
would then have $Y\dep{\E,\S} X$, contradicting the assumptions. Then we can 
compute:
\allowdisplaybreaks
\begin{align*}
  & \sem{\E,\S} \\
=\; & \mbox{~~~~~~(Theorem~\ref{splitsolve})}\\
  & \sem{\E,L_1\conc L_3\conc\sigma Y\conc L_4}\\
=\; & \mbox{~~~~~~(Theorem~\ref{unfold})}\\
  & \sem{\unfold(\E,X,Y),L_1\conc L_3\conc\sigma Y\conc L_4}\\
=\; & \mbox{~~~~~~(Theorem~\ref{splitsolve}, observation (2) above)}\\
  & \sem{\unfold(\E,X,Y),\S}
\qedhere
\end{align*}
\end{proof}

\section{Summary -- Examples -- Related Work}
\label{related}
Table~\ref{mainresults} summarizes our main results. We will discuss their
relevance and compare them to previous work in 
Section~\ref{sub:subst}-\ref{sub:swap}. Table~\ref{otherresults}
contains some other useful facts on FES, discussed in Section~\ref{sub:other}.

\subsection{Substituting Definitions and Solutions}\label{sub:subst}
Theorem~\ref{unfold} in this form is new. It generalizes
\cite[Lemma~6.3]{Mader96} (for BES only) and
\cite[Lemma~18]{GrWiTCS05} (for PBES only) to FES.
Another generalization is that we allow that
$X=Y$. That is, besides unfolding the $Y$'s in the definition of some
$X$ {\em preceding} $Y$, one can even unfold $Y$ in its own
definition. The proof for this case is more involved
(cf. Lemma~\ref{unfoldsame}). For BES this is useless, but for PBES
this is useful, and already used in~\cite{Orzan2010,Ploeger2011} to unfold PBESs
to BESs.
The technique of unfolding PBESs is perhaps the most commonly applied
method of solving PBESs~\cite{faisal_al_ameen_asynchronous_2024,kobayashi_foldunfold_2020,Ploeger2011}, although symbolic approaches do exist~\cite{kobayashi_temporal_2019,NeeleWG2020}.

Theorem~\ref{unfoldloop} is a new result, generalizing the case where
$\indep(\E,Y,X)$ for all $X$ (i.e.~$Y$ is in solved form,
\cite{Mader96,GrWiTCS05}). In general, one cannot unfold $Y$ in the
definition of $X$, when $Y$ precedes $X$. However, if there is no
dependency path from $Y$ to $X$, then a forward substitution is
allowed. The proof is based on clever reordering of equations.
The following example shows that this condition is necessary:

\begin{exa}\label{exunfold} 
	Consider the following two Boolean Equation Systems:
	\[
	\begin{array}{|c|c|}
		\hline
		B_1 & B_2\\
		\hline
		\begin{array}{l}
			\nu Y=X\\
			\mu X=Y\\
		\end{array} &
		\begin{array}{c}
			\nu Y=X\\
			\mu X=X\\
		\end{array}\\
		\hline
	\end{array}
	\]
	Unfolding $Y$ in the definition of $X$ in $B_1$ yields $B_2$. However, $B_1$ has
	the solution $(\top,\top)$, while the solution of $B_2$ is $(\bot,\bot)$.
	The reader can check this with the method described in Example~\ref{gauss}.
\end{exa}

%% Note that subsituting the definition of $Y$ in $X$ is allowed, but has
%% no effect in case $X\not\rightarrow Y$.

Theorem~\ref{partial} allows to substitute a partial solution in a
FES.  It occurs already in \cite[Lemma 3.19]{Mader96}. However, our
proof is more direct. Mader suggests that a direct inductive argument
is not possible, and proves the theorem by contradiction, constructing
an infinite set of equation systems. We show that with an appropriate
induction loading, the theorem can be reduced to another lemma in
fixpoint calculus (Lemma~\ref{fixcalculus}.8).

The substitution theorems form the basis for solving BES and PBES by
Gauss-elimination. They are called the global steps. Besides global
steps, one needs local steps, to eliminate $X$ from the right hand
side of its own definition. For BES, a local step is trivial, because
(only) in the Boolean lattice we have $\mu X.f(X) = f(\bot)$ and $\nu
X.f(X)=f(\top)$.  Local solution for PBES is much harder, and studied
in~\cite{GrWiTCS05,Orzan2010}. We stress that our results show that the
global steps hold in any FES. However, effective local solution is
specific to the underlying complete lattice.

\begin{table}
	\begin{center}
		\begin{tabular}{|l|c|l|}
			\hline
			{\bf Thm} & {\bf (In)Equality} & {\bf Conditions}\\
			\hline\hline
			\multicolumn{3}{|c|}{{\bf Reordering Variables}}\\
			\hline
			\ref{swapsame} &
			$\sem{\E,\S_1\conc\sigma X\conc\sigma Y\conc\S_2}=
			\sem{\E,\S_1\conc\sigma Y\conc\sigma X\conc\S_2}$ &
			- $\E$ is monotonic\\
			\hline
			\ref{migineq} &
			$\sem{\E,\S_1\conc\mu X\conc\nu Y\conc\S_2}\leq
			\sem{\E,\S_1\conc\nu Y\conc\mu X\conc\S_2}$ &
			- $\E$ is monotonic\\
			&& - $X\neq Y$\\
			\hline
			\ref{swaploop} &
			$\sem{\E,\S_1\conc\sigma X\conc\rho Y\conc\S_2} 
			= \sem{\E,\S_1\conc\rho Y\conc\sigma X\conc\S_2}$ &
			- $X\not\dep{} Y$ in $(\E,\sigma X\conc\rho Y\conc\S_2)$\\
			\hline
			\ref{migration2} &
			${\sem{\E,\S_0\conc\S_1\conc\S_2\conc\S_3} =
				\sem{\E,\S_0\conc\S_2\conc\S_1\conc\S_3}}$ &
			- $\disjoint(\S_1, \S_2\conc\S_3)$\\
			&& - either $\indep(\E,\S_1,\S_2\conc\S_3)$,\\
			&& ~~ or $\indep(\E,\S_2\conc\S_3,\S_1)$\\
			\hline
			\multicolumn{3}{|c|}{{\bf Substituting Definitions and Solutions}}\\
			\hline
			\ref{unfold} & 
			$\sem{\unfold(\E,X,Y),\S}=\sem{\E,\S}$ &
			- $\E$ is monotonic\\
			&& - $\S=\S_1\conc\sigma Y\conc\S_2$\\
			& ($X=Y$ is allowed)& - $X\notin\dom(\S_2)$\\
			\hline
			\ref{unfoldloop} &
			$\sem{\unfold(\E,X,Y),\S}=\sem{\E,\S}$ &
			- $\E$ is monotonic\\
			&& - $\S=\S_1\conc\sigma Y\conc\S_2$\\
			&& - $Y\not\dep{} X$ in $(\E,\S)$\\
			\hline
			%\ref{unfold} & 
			%$\sem{\E\cup\{X=E_X,Y=E_Y\}),\S}=\sem{\E\cup\{X=E_X[Y:=E_Y],Y=E_Y\},\S}$ &
			%   - $\E$ is monotonic\\
			%&& - $X$ not after $Y$ in $\S$\\
			%\hline
			\ref{partial} &
			$\sem{\E,\S}(\eta) = \sem{\E[X\mapsto A],\S}(\eta)$ &
			- $\E$ is monotonic\\
			&& - $A = \sem{\E,\S}(\eta)(X)$\\
			\hline
			\multicolumn{3}{|c|}{{\bf Swapping Signs}}\\
			\hline
			\ref{signineq} & 
			$\sem{\E,\S_1\conc\mu X\conc\S_2} \leq \sem{\E,\S_1\conc\nu X\conc\S_2}$& 
			- $\E$ is monotonic\\
			\hline
			\ref{signloop} &
			$\sem{\E,\S_1\conc\mu X\conc\S_2} = \sem{\E,\S_1\conc\nu X\conc\S_2}$ & 
			- $X\notin\dom(\S_2)$\\
			&& - $X\not\depnonempty{} X$ in $(\E,\mu X\conc\S_2)$\\
			\hline
		\end{tabular}
	\end{center}
	\caption{Main results for arbitrary FES\label{mainresults}}
\end{table}

\begin{exa}
	\label{gauss}
	The following example shows the solution of a BES by Gauss
	elimination.  Basically, one first substitutes definitions backwards
	using Theorem~\ref{unfold} (along the way, we use the identity
	$Y\vee(Y\wedge X) \equiv Y$):
	\[\begin{array}{|rcl|}
		\hline
		\mu X & = & Y\vee Z\\
		\nu Y & = & Z\\
		\mu Z & = & Y\wedge X\\
		\hline
	\end{array}
	~\to~
	\begin{array}{|rcl|}
		\hline
		\mu X & = & Y\\
		\nu Y & = & Y\wedge X\\
		\mu Z & = & Y\wedge X\\
		\hline
	\end{array}
	~\to~
	\begin{array}{|rcl|}
		\hline
		\mu X & = & Y\wedge X\\
		\nu Y & = & Y\wedge X\\
		\mu Z & = & Y\wedge X\\
		\hline
	\end{array}
	\]
	Next, one obtains $X=\bot$ by a local elimination step in the first equation,
	using that $Y\wedge\bot \equiv \bot$.
	This solution can then be substituted forward by Theorem~\ref{unfoldloop},
	to obtain the full solution $(\bot,\bot,\bot)$. In general, steps 1 and 2
	must be mixed.
\end{exa}

The next example shows a PBES where unfolding $X$ in its own
definition makes sense. 
\begin{exa}
	Applying Theorem~\ref{unfold} to unfold $X$ in its own definition, we get:
	\[\begin{array}{|rcl|}
		\hline
		\nu Y    & \!\!=\!\! & X(\top) \\
		\mu X(b) & \!\!=\!\! & (b\wedge Y) \vee X(\neg b)\\
		\hline
	\end{array}
	~\to~ 
	\begin{array}{|rcl|}
		\hline
		\nu Y    & \!\!=\!\! & X(\top) \\
		\mu X(b) & \!\!=\!\! & (b\wedge Y) \vee ((\neg b\wedge Y) \vee X(\neg\neg b))\\
		& \!\!\equiv\!\! & Y\vee X(b)\\
		\hline
	\end{array}\]
	Applying Theorem~\ref{unfold} again, to unfold $X$ in $Y$ yields
	$\nu Y = Y\vee X(\top)$, hence by local resolution $Y=\top$, hence $X(b)=\top$.
\end{exa}

\subsection{Reordering Variables}\label{sub:reorder}
Theorem~\ref{swapsame} indicates that two adjacent variables with the
same sign may be interchanged. This theorem occurs already in
\cite[Lemma 3.21]{Mader96}.  For PBES it is repeated in \cite[Lemma
21]{GrWiTCS05}. However, \cite{Mader96,GrWiTCS05} don't give a full
proof, but refer to Beki\v{c} Lemma. In our proof, we show exactly how
Theorem~\ref{unfold} reduces to our version of Beki\v{c} Rule
(Lemma~\ref{fixcalculus}.9).
In other works~\cite{Seidl96,kobayashi_temporal_2019}, adjacent variables
with the same sign are grouped in unordered \emph{blocks}.
No claim is made about the correctness of such a definition.

Theorem~\ref{migineq} shows the inequality that arises when
interchanging adjacent variables with different sign. It occurs
already in \cite[Lemma 3.23]{Mader96}, but our proof is different. Our
proof depends on a probably new inequality in fixpoint calculus, which
we coin Beki\v{c} Inequality (Lemma~\ref{fixinequalities}.4).

Theorem~\ref{swaploop} is a new result. It states that in the special
case that $X$ and $Y$ are not on the same dependency loop, they can be
interchanged without modifying the solution. This generalizes
\cite[Lemma 19]{GrWiTCS05}, which requires that the right-hand side of
$Y$ in $\E$ is a constant, \emph{i.e.}, $\indep(\E,Y,Z)$ for all $Z$.

Finally, Theorem~\ref{migration2} in this form is new. It allows to
swap whole blocks of equations. Mader~\cite[Lemma 3.22]{Mader96}
claims a similar result, under the condition that both
$\indep(\E,\S_1,\S_2)$ and $\indep(\E,\S_2,\S_1)$. However,
\cite{GrWiTCS05} show a counter example to this. The repair
in~\cite[Lemma 22]{GrWiTCS05} requires that $\S_3$ is empty. We show
a stronger result: if $\S_3$ is empty, only one of the requirements
$\indep(\S_2,\S_1)$ or $\indep(\S_2,\S_1)$ is needed.

Notably, our result even applies to nonempty $\S_3$, provided we have
$\indep(\S_1,\S_2\conc\S_3)$ (or its reverse), i.e. $\S_1$ is also
independent on the variables in $\S_3$. Note that we allow arbitrary
(dependent) alternations within $\S_1$ and $\S_2$, and even $\S_2$
might depend on $\S_1$. 
We lifted two other unnecessary restrictions:
surprisingly, this result doesn't require monotonicity of $\E$. Also,
the results in~\cite{Mader96,GrWiTCS05} are for individual equations
only, while we can swap whole blocks at the same time. 

We now show
an application of swapping blocks to reduce the number of $\mu/\nu$-alternations.

\begin{exa}
	Consider the following four Boolean Equation Systems:
	\[
	\begin{array}{|c|c|c||c|}
		\hline
		B_3 & B_4 & B_5 & B_6 \\
		\hline
		\begin{array}{rcl}
			\mu X & = & Y \\
			\mu Y & = & X \\
			\nu Z & = & W \\
			\mu W & = & Z \\
		\end{array} &
		\begin{array}{rcl}
			\mu X & = & Y \\
			\nu Z & = & W \\
			\mu Y & = & X \\
			\mu W & = & Z \\
		\end{array} &
		\begin{array}{rcl}
			\nu Z & = & W \\
			\mu X & = & Y \\
			\mu Y & = & X \\
			\mu W & = & Z \\
		\end{array} &
		\begin{array}{rcl}
			\mu X & = & Y \\
			\mu Y & = & X \\
			\mu W & = & Z \\
			\nu Z & = & W \\
		\end{array}\\
		\hline
	\end{array}
	\]
	
	For these BES, the dependency graph between the variables consists
	of two loops, $X\leftrightarrow Y$ and $Z\leftrightarrow W$.
	In particular, we have $\indep(\{X,Y\},\{Z,W\})$.
	We want to transform $B_3$ to $B_5$, because it has fewer alternations.
	Theorem~\ref{swapsame} cannot be applied, because the sign of $Z$ is
	different from all the others.
	
	Using Theorem~\ref{migration2} on individual equations, one can show
	that $\sem{B_3}=\sem{B_4}$, because $\indep(Y,\{Z,W\})$.  However, one
	cannot derive $\sem{B_4}=\sem{B_5}$ using Theorem~\ref{migration2},
	because neither $\indep(X,\{Z,Y,W\})$, nor $\indep(\{Z,Y,W\},X)$
	holds.  However, one can prove $B_3=B_5$ directly with
	Theorem~\ref{migration2}, by swapping block $[X,Y]$ with $Z$, because
	indeed we have $\indep(\{X,Y\},\{Z,W\})$. Alternatively, one can
	observe that $X\not\dep{} Z$, and apply Theorem~\ref{swaploop}
	to deduce that $B_4=B_5$ directly.
	
	All theorems fail to prove the equivalence of $B_{3-5}$ with $B_6$.
	However, Theorem~\ref{migineq} guarantees that $\sem{B_6}\leq\sem{B_3}$.
	As a matter of fact, the solution of $B_3$, $B_4$ and $B_5$ is
	$(X=\bot,Y=\bot,Z=\top,W=\top)$, while the solution of $B_6$ is
	$(X=\bot,Y=\bot,Z=\bot,W=\bot)$. The reader may verify this by
	Gauss Elimination, cf.~Example~\ref{gauss}. 
	This shows that the reordering theorems cannot easily be strengthened.
\end{exa}

%% The following also shows the need for a generalized migration theorem
%% \begin{verbatim}
	%% mu x = y /\ z          nu y = z
	%% nu y = z         <==>  mu x = y /\ z
	%% mu z = y               mu z = y
	%% \end{verbatim}

%% Also check (but I completely forgot why??):
%% \begin{verbatim}
	%% mu x = z \/ y
	%% nu y = z
	%% nu z = y \/ x
	%% \end{verbatim}

\subsection{Swapping Signs}\label{sub:swap}
The inequality of Theorem~\ref{signineq} is well known and appears for
instance in \cite[Lemma 3.24]{Mader96}. Theorem~\ref{signloop} is new. It shows
that the sign of variable $X$ is only relevant when $X$ is the most
relevant variable on a dependency loop.

\begin{exa}
	Now consider the next three BESs which only differ in their fixpoint signs:
	\[
	\begin{array}{|c|c|c|}
		\hline
		B_7 & B_8 & B_9\\
		\hline
		\begin{array}{rcl}
			\mu X & = & Y \\
			\nu Y & = & X \lor Z \\
			\mu Z & = & Z \land W \\
			\nu W & = & X \land \bot \\
		\end{array} &
		\begin{array}{rcl}
			\mu X & = & Y \\
			\nu Y & = & X \lor Z \\
			\nu Z & = & Z \land W \\
			\nu W & = & X \land \bot \\
		\end{array} &
		\begin{array}{rcl}
			\mu X & = & Y \\
			\mu Y & = & X \lor Z \\
			\mu Z & = & Z \land W \\
			\mu W & = & X \land \bot \\
		\end{array}\\
		\hline
	\end{array}
	\]
	Like before, want to manipulate $B_7$ to reduce the number of fixpoint alternations.
	We identify two possibilities.
	First, we may flip the sign of $Z$, obtaining $B_8$, which has the solution $(\top,\top,\bot,\bot)$.
	By Theorem~\ref{signineq}, it holds that $\sem{B_7} \leq \sem{B_8}$ and so we conclude that also $\sem{B_7}(Z) = \sem{B_7}(W) = \bot$.
	The other option is to flip the sign of $Y$ and $W$, yielding $B_9$.
	Since $Y \not\depnonempty{} Y$ in $Y,Z,W$, Theorem~\ref{signloop} gives that $\sem{B_7} = \sem{B_9}$.
\end{exa}

\subsection{Other Results}\label{sub:other}
Along the way, we proved (and proof checked) several lemmas on FES
that may be interesting on their own. For quick reference, we
summarize these in Table~\ref{otherresults}. Here $\eta$ denotes an
arbitrary valuation. All these lemmas occur in some form in the literature.
Lemma~\ref{solution} states that the semantics indeed returns a solution, and
follows from \cite[Lemma 3.5]{Mader96}.  Lemma~\ref{indepsolve} corresponds to
\cite[Lemma 3.10]{Mader96} (which is not proved there) and \cite[Lemma
7]{GrWiTCS05}. Actually, \cite{Mader96} has Lemma~\ref{indepsolve2},
which is equivalent according to our
Theorem~\ref{migration2}. Lemma~\ref{congruence} and
\ref{left_ineq_cong} are from~\cite[Lemma~3.14]{Mader96} as well, and
Lemma~\ref{congright} follows directly from Lemma~\ref{indepsolve}.
We included it here to stress that right congruence doesn't hold in
general.

\begin{table}
	\begin{center}
		\begin{tabular}{|l|c|l|}
			\hline
			{\bf Lem} & {\bf Result} & {\bf Condition} \\
			\hline
			\hline
			\ref{solution} & 
			$\sem{\E,\S}(\eta)(X) = \E_X(\sem{\E,\S}(\eta))$&
			- $\E$ monotonic\\
			&& - $\X\in\dom(\S)$\\
			\hline
			\ref{indepsolve} & 
			$\sem{\E,\S_1\conc\S_2}(\eta) = \sem{\E,\S_2}(\sem{\E,\S_1}(\eta))$
			& - $\indep(\E,\S_1,\S_2)$\\
			\hline
			\ref{indepsolve2} & 
			$\sem{\E,\S_1\conc\S_2}(\eta) = \sem{\E,\S_1}(\sem{\E,\S_2}(\eta))$
			& - $\indep(\E,\S_2,\S_1)$\\
			&& - $\disjoint(\S_1,\S_2)$\\
			\hline
			\hline
			\ref{left_ineq_cong} & 
			$\sem{\E,\S_1}\leq \sem{\E,\S_2}$ 
			implies $\sem{\E,\S\conc\S_1}\leq\sem{\E,\S\conc\S_2}$&
			- $\E$ monotonic\\
			\hline
			\ref{congruence} & 
			$\sem{\E,\S_1}=\sem{\E,\S_2}$ 
			implies $\sem{\E,\S\conc\S_1}=\sem{\E,\S\conc\S_2}$&\\
			\hline
			\ref{congright} & $\sem{\E,\S_1}=\sem{\E,\S_2}$ 
			implies $\sem{\E,\S_1\conc\S}=\sem{\E,\S_2\conc\S}$ &
			- $\indep(\E,\S_1,\S)$\\
			&& - $\indep(\E,\S_2,\S)$\\
			&& - $\disjoint(\S,\S_1\conc\S_2)$\\
			\hline
		\end{tabular}
	\end{center}
	\caption{Some useful lemmas for arbitrary FES\label{otherresults}}
\end{table}

Finally, we needed some basic results on fixpoints in complete
lattices, cf.~Lemma~\ref{fixcalculus} and~\ref{fixinequalities}. The
existence and definition of least and greatest fixpoints is due to
Knaster (on sets) and Tarski (on complete lattices)~\cite{TARS55},
see~\cite{LNS82} for a historical account. We (re)proved a number of
identities (Lemma~\ref{fixcalculus}) and inequalities
(Lemma~\ref{fixinequalities}) on fixpoint expressions. Most of
these results are known.  Lemma~\ref{fixcalculus}.1-6
can for instance be found in~\cite{Backhouse02}. Rule 9 (Beki\v{c} Equality)
can be found in e.g.~\cite{Bakker80,Bekic84}, but stated in a different form,
involving simultaneous fixpoints. We have not found in the literature
the inequality in Lemma~\ref{fixinequalities}.4, which resembles
Beki\v{c} equality on terms with mixed minimal and maximal fixpoints.

\section{Formalisation in Coq \& PVS}
\label{coqpvs}
We have formalised all of the above theory in both Coq~\cite{bertot_short_2008,coq_software} and PVS~\cite{mohamed_brief_2008}.
A replication artefact containing these proofs is available at~\cite{artefact}.
The formalized definitions and proofs follow the definitions and proof steps
in this paper quite closely.
Here, we highlight the main difference between the two formalisations.

In Coq, we captured the concepts of complete lattices and monotonic functions in typeclasses.
For these, we defined several typeclass instances, for example the product lattice and composition of monotonic functions.
In many cases, Coq is able to perform automatic typeclass resolution, saving us from manually proving monotonicity of complex functions, for example those in Lemma~\ref{completelattices}.
Furthermore, Coq supports user-defined notation, allowing us to closely follow the notation used in the paper.
The proofs for showing decidability of $X \dep{\E,\S} Y$ are extensive, something that is not reflected in the paper.

Our PVS definitions and proofs were originally developed under PVS version 4.2, but could be ported to version 7.1 with minimal effort.
Contrary to Coq, PVS is built on classical logic and thus allows the law of excluded middle (for all propositions $P$, it holds $P \lor \neg P$).
We thus do not need to supply proofs for decidability of $X \dep{\E,\S} Y$.
This also means that we do not rely on finiteness of $\S$, and thus the definition of $\dep{}$ only depends on $\E$ and the domain $\E$ is restricted where necessary, \emph{e.g.}, in Theorem~\ref{signloop}.
This simplifies the proof of Lemma~\ref{splitsolve_aux}: it can operate on $\splitname_{X,\E}$ directly.

\section{Conclusion}
We provided several equalities and inequalities involving
a range of operations on fixpoint equation
systems (FES). We refer to Table~\ref{mainresults} and~\ref{otherresults}
(Section~\ref{related}) for a summary of the theorems.
Lemmas~\ref{fixcalculus} and~\ref{fixinequalities} provide
a useful overview on equalities
and inequalities for nested fixed points in complete lattices.

We provided self-contained and detailed proofs of all results
and mechanised these proofs in two proof assistants, Coq and PVS.

By the generic nature of FES, these results carry over to other
formalisms such as Boolean equation systems (BES), 
parity games (and variations thereof), and parameterised (first-order)
Boolean equation systems (PBES). 

\section*{Acknowledgment}
%The authors are grateful to Jaap van der Woude for his help on
%fixpoints; to Jan Friso Groote and Tim Willemse for their
%help on PBESs; and to Wan Fokkink, Alban Ponse, Jan Rutten and
%Michael Weber for useful discussions. 
Large part of the research
was carried out at the Centrum voor Wiskunde en Informatica. 

\bibliographystyle{alphaurl}
\bibliography{pbes}

\end{document}